\documentclass[a4paper,11pt]{article}
\usepackage[top=2.7cm,bottom=3cm,left=2.5cm,right=2.5cm,footskip=.7in]{geometry}

\usepackage{amsmath}
\usepackage{amssymb}
\usepackage{amsfonts,amsthm,complexity}
\usepackage[hidelinks]{hyperref}
\newtheorem{theorem}{Theorem}%[section]
\newtheorem{lemma}[theorem]{Lemma}
\newtheorem{claim}{Claim}[theorem]
\newtheorem{observation}[theorem]{Observation}

\newtheorem{corollary}[theorem]{Corollary}
\theoremstyle{definition}
\newtheorem*{definition}{Definition}
\usepackage{thm-restate}
\usepackage{todonotes}

\usepackage{graphicx}
\usepackage{colortbl}
\usepackage{multicol}
\usepackage[normalem]{ulem}
\usepackage{authblk}

\usepackage{tikz}
\usetikzlibrary{snakes}
\usetikzlibrary{decorations.pathmorphing,matrix}

\newcommand\mcut{\textsc{Max Cut }}
\newcommand{\mypath}[1]{\langle#1\rangle}
\newcommand{\optdir}[2]{A_{#1}^{#2}}

\newenvironment{claimproof}[1]{\par\noindent\emph{Proof:}\space#1}{{\leavevmode\unskip\penalty9999 \hbox{}\nobreak\hfill\quad\hbox{$\lrcorner$}}\medskip}

\title{Computing Subset Feedback Vertex Set via Leafage\thanks{Research supported by the Hellenic Foundation for Research and Innovation (H.F.R.I.) under the ``First Call for H.F.R.I. Research Projects to support Faculty members and Researchers and the procurement of high-cost research equipment grant'', Project FANTA (eFficient Algorithms for NeTwork Analysis), number HFRI-FM17-431.}}
\author[1]{Charis Papadopoulos}
\author[2]{Spyridon Tzimas}

\affil[1]{Department of Mathematics, University of Ioannina, Greece\\
 \texttt{charis@uoi.gr}}
\affil[2]{Department of Mathematics, University of Ioannina, Greece\\
 \texttt{roytzimas@hotmail.com}}

\date{}

\begin{document}
\maketitle

\begin{abstract}
Chordal graphs are characterized as the intersection graphs of subtrees in a tree and such a representation is known as the tree model.
Restricting the characterization results in well-known subclasses of chordal graphs such as interval graphs or split graphs.
A typical example that behaves computationally different in subclasses of chordal graph is the \textsc{Subset Feedback Vertex Set} (SFVS) problem:
given a vertex-weighted graph $G=(V,E)$ and a set $S\subseteq V$, the \textsc{Subset Feedback Vertex Set} (SFVS) problem asks for a vertex set of minimum weight that intersects all cycles containing a vertex of $S$.
SFVS is known to be polynomial-time solvable on interval graphs, whereas SFVS remains \NP-complete on split graphs and, consequently, on chordal graphs.
Towards a better understanding of the complexity of SFVS on subclasses of chordal graphs,
we exploit structural properties of a tree model in order to cope with the hardness of SFVS.
Here we consider variants of the \emph{leafage} that measures the minimum number of leaves in a tree model.
We show that SFVS can be solved in polynomial time for every chordal graph with bounded leafage.
% By combining recent algorithms related to mim-width, there is an algorithm for SFVS that, given a chordal graph on $n$ vertices with leafage $\ell$, runs in $n^{O(\ell^2)}$ time.
In particular, given a chordal graph on $n$ vertices with leafage $\ell$, we provide an algorithm for SFVS with running time $n^{O(\ell)}$.
%% and thus improving upon the $n^{O(\ell^2)}$-time algorithm via mim-width.
%% showing that SFVS is in XP paramerized by $\ell$.
%%
% should we mention a $n^{O(\ell^2)}$-time algorithm is achieved through the mim-width approach?
We complement our result by showing that SFVS is \W[1]-hard parameterized by $\ell$.
Pushing further our positive result, it is natural to consider a slight generalization of leafage, the \emph{vertex leafage}, which measures the
smallest number among the maximum number of leaves of all subtrees in a tree model.
However, we show that it is unlikely to obtain a similar result, as we prove that SFVS remains \NP-complete on undirected path graphs, i.e., graphs having vertex leafage at most two.  %% showing that SFVS is para-\NP-complete parameterized by the vertex leafage.
Moreover, we strengthen previously-known polynomial-time algorithm for SFVS on rooted path graphs that form a proper subclass of undirected path graphs and graphs of mim-width one.
\end{abstract}

%\keywords{Subset feedback vertex set; leafage; intersection graphs}

\section{Introduction}
Several fundamental optimization problems are known to be intractable on chordal graphs, however they admit polynomial time algorithms when restricted to a proper subclass of chordal graphs such as interval graphs.
Typical examples of this type of problems are domination or induced path problems \cite{BB86,BH82,CorneilP84,HeggernesHLS15,IoannidouMN11,NatarajanS96}.
%;
%see also the books \cite{graph:classes:brandstadt:1999,graph:classes:Go04,McKeeMcMor99,Spinrad03} for an overview and a thorough exposition.
%%%cite: domination chordal \cite{BH82,CorneilP84} and longest induced path \cite{BB86,IoannidouMN11} and disjoint induced path \cite{NatarajanS96,HeggernesHLS15}
Towards a better understanding of why many intractable problems on chordal graphs admit polynomial time algorithms on interval graphs, we consider the algorithmic usage of the structural parameter named leafage.
%% among others
Leafage, introduced by Lin et al. \cite{LinMW98}, is a graph parameter that captures how close is a chordal graph of being an interval graph.
As it concerns chordal graphs, leafage essentially measures the smallest number of leaves in a clique tree, an intersection representation of the given graph \cite{G74}.
Here we are concerned with the \textsc{Subset Feedback Vertex Set} problem, SFVS for short:
given a vertex-weighted graph and a set $S$ of its vertices, compute a vertex set of minimum weighted size that intersects all cycles containing a vertex of $S$.
Although \textsc{Subset Feedback Vertex Set} does not fall to the themes of domination or induced path problems, it is known to be \NP-complete on chordal graphs \cite{FominHKPV14}, whereas it becomes polynomial-time solvable on interval graphs \cite{PapT19}.
Thus our research study concerns to what extent the structure of the underlying tree representation influences the computational complexity of \textsc{Subset Feedback Vertex Set}.

An interesting remark concerning \textsc{Subset Feedback Vertex Set}, is the fact that its unweighted and weighted variants behave computationally different on hereditary graph classes.
For example, \textsc{Subset Feedback Vertex Set} is \NP-complete on $H$-free graphs for some fixed graphs $H$, while its unweighted variant admits polynomial time algorithm on the same class of graphs \cite{abs-2007-14514,PapadopoulosT20}.
Thus the unweighted and weighted variants of \textsc{Subset Feedback Vertex Set} do not align.
This comes in contrast even to the original problem of \textsc{Feedback Vertex Set} which is obtained whenever $S=V(G)$.
\textsc{Subset Feedback Vertex Set} remains \NP-complete on bipartite graphs \cite{Yannakakis81a} and planar graphs~\cite{GJ}, as a generalization of \textsc{Feedback Vertex Set}.
Notable differences between the two latter problems regarding their complexity status is the class of split graphs and $4P_1$-free graphs
for which \textsc{Subset Feedback Vertex Set} is \NP-complete \cite{FominHKPV14,PapadopoulosT20}, as opposed to the \textsc{Feedback Vertex Set} problem \cite{fvs:chord:corneil:1988,Spinrad03,abs-2007-14514}.
%% lemma 16 in abs-2007-14514 shows that weighted FVS is poly on 4P_1-free graphs.
Inspired by the \NP-completeness on chordal graphs, \textsc{Subset Feedback Vertex Set} restricted on (subclasses of) chordal graphs has attracted several researchers to obtain faster, still exponential-time, algorithms \cite{GolovachHKS14,PhilipRST19}.

On the positive side, \textsc{Subset Feedback Vertex Set} can be solved in polynomial time on restricted graph classes \cite{abs-2007-14514,BJPPwg20,PapT19,PapadopoulosT20}.
%Among the stated classes, we stress that interval graphs is the only known subclass of chordal graphs for which \textsc{Subset Feedback Vertex Set} is solved in polynomial time.
Related to the structural parameter mim-width, Bergougnoux et al. \cite{BergPTwg20} recently proposed an $n^{O(w^2)}$-time algorithm that solves \textsc{Subset Feedback Vertex Set}
given a decomposition of the input graph of mim-width $w$.
As leaf power graphs admit a decomposition of mim-width one \cite{JaffkeKST19}, from the later algorithm
\textsc{Subset Feedback Vertex Set} can be solved in polynomial time on leaf power graphs if an intersection model is given as input.
%%that form a subclass of chordal graphs
However, to the best of our knowledge, it is not known whether the intersection model of a leaf power graph can be constructed in polynomial time.
Moreover, even for graphs of mim-width one that do admit an efficient construction of the corresponding decomposition,
the exponent of the running time given in \cite{BergPTwg20} is relatively high.

%A natural parameter to measure the structure of the tree representation is the leafage of the given chordal graph.
Habib and Stacho \cite{HabibS09} showed that the leafage of a connected chordal graph can be computed in polynomial time.
Their described algorithm also constructs a corresponding clique tree with the minimum number of leaves.
Regarding other problems that behave well with the leafage, we mention the \textsc{Minimum Dominating Set} problem for which Fomin et al. \cite{FominGR20} showed that the problem is \FPT{} parameterized by the leafage of the given graph.
Here we show that \textsc{Subset Feedback Vertex Set} is polynomial-time solvable for every chordal graph with bounded leafage.
In particular, given a chordal graph with a tree model having $\ell$ leaves, our algorithm runs in $O(\ell n^{2\ell+1})$ time.
Thus, by combining the algorithm of Habib and Stacho \cite{HabibS09}, we deduce that \textsc{Subset Feedback Vertex Set} is in \XP, parameterized by the leafage.

One advantage of leafage over mim-width is that we can compute the leafage of a chordal graph in polynomial time, whereas we do not know how to compute in polynomial time the mim-width of a chordal graph.
However we note that a graph of bounded leafage implies a graph of bounded mim-width and, further, a decomposition of bounded mim-width can be computed in polynomial time \cite{FominGR20}.
This can be seen through the notion of $H$-graphs which are exactly the intersection graphs of connected subgraphs of some subdivision of a fixed graph $H$.
The intersection model of subtrees of a tree $T$ having $\ell$ leaves is a $T'$-graph where $T'$ is obtained from $T$ by contracting nodes of degree two.
Thus the size of $T'$ is at most $2 \ell$, since $T$ has $\ell$ leaves.
Moreover, given an $H$-graph and its intersection model, a (linear) decomposition of mim-width at most $2 |E(H)|$ can be computed in polynomial time \cite{FominGR20}.
Therefore, given a graph of leafage $\ell$, there is a polynomial-time algorithm that computes a decomposition of mim-width $O(\ell)$.
Combined with the algorithm via mim-width \cite{BergPTwg20}, one can solve \textsc{Subset Feedback Vertex Set} in time $n^{O({\ell}^{2})}$ on graphs having leafage $\ell$.
Notably, our $n^{O(\ell)}$-time algorithm is a non-trivial improvement on the running time obtained from the mim-width approach.

We complement our algorithmic result by showing that \textsc{Subset Feedback Vertex Set} is W[1]-hard parameterized by the leafage of a chordal graph.
Thus we can hardly avoid the dependence of the exponent in the stated running time.
Our reduction is inspired by the W[1]-hardness of \textsc{Feedback Vertex Set} parameterized by the mim-width given in \cite{JaffkeKT20}.
However we note that our result holds on graphs with arbitrary vertex weights and we are not unaware if the unweighted variant of \textsc{Subset Feedback Vertex Set} admits the same complexity behavior.

%%% Here speak about W[]-hardness....
%\todo[inline]{Here speak about W[]-hardness....}

Our algorithm works on an expanded tree model that is obtained from the given tree model and maintains all intersecting information without increasing the number of leaves.
Then in a bottom-up dynamic programming fashion, we visit every node of the expanded tree model in order to compute partial solutions.
At each intermediate step, we store all necessary information of subsets of vertices that are of size $O(\ell)$.
As a byproduct of our dynamic programming scheme and the expanded tree model,
we show how our approach can be extended in order to handle rooted path graphs.
Rooted path graphs are the intersection graphs of rooted paths in a rooted tree.
They form a subclass of leaf powers and have unbounded leafage (through their underlying tree model).
Although rooted path graphs admit a decomposition of mim-width one \cite{JaffkeKST19} and such a decomposition can be constructed in polynomial time \cite{Dietz84,Gavril75}, %% here to be corrected!!
the running time obtained through the bounded mim-width approach is rather unpractical, as it requires to store a table of size $O(n^{13})$ even in this particular case \cite{BergPTwg20}.
%% p.21 (p.18 arxiv): |I| < n^{8+5}
By analyzing further subsets of vertices at each intermediate step, we manage to derive an algorithm for \textsc{Subset Feedback Vertex Set} on rooted path graphs that runs in $O(n^2m)$ time.
Observe that the stated running time is comparable to the $O(nm)$-time algorithm on interval graphs \cite{PapT19} and interval graphs form a proper subclass of rooted path graphs.

%\begin{figure}[t]
%\centering
%\includegraphics[scale= 0.74]{fig-overview-new.pdf}
%\caption{Complexity of SFVS parameterized by leafage and vertex leafage.}\label{fig:leafage}
%\end{figure}

Moreover, inspired by the algorithm on bounded leafage graphs we consider its natural generalization concerning the \emph{vertex leafage} of a graph.
Chaplick and Stacho \cite{ChaplickS14} introduced the vertex leafage of a graph $G$ as the smallest number $k$ such that there exists a tree model for $G$ in which every subtree corresponding to a vertex of $G$ has at most $k$ leaves.
As leafage measures the closeness to interval graphs (graphs with leafage at most two), vertex leafage measures the closeness to undirected path graphs which are the intersection graphs of paths in a tree (graphs with vertex leafage at most two).
We prove that the unweighted variant of \textsc{Subset Feedback Vertex Set} is \NP-complete on undirected path graphs and, thus, the problem is para-\NP-complete parameterized by the vertex leafage.
An interesting remark of our \NP-completeness proof is that our reduction comes from the \textsc{Max Cut} problem as opposed to known reductions for
\textsc{Subset Feedback Vertex Set} which are usually based on, more natural, covering problems \cite{FominHKPV14,PapadopoulosT20}.
Thus we obtain a complexity dichotomy of the problem restricted on the two comparable classes of rooted and undirected path graphs.
Our findings are summarized in Figure~\ref{fig:leafage}.
%\footnote{Our findings are summarized in Figure~\ref{fig:leafage} given in Appendix~\ref{sec:findings}.}.
%% We should mention that vertex leafage concerns even the rooted subtrees obtained from the given (undirected) tree.

%% remember: W[1]-hardness on weighted chordal graphs, and NP-hardness on unweighted undirected path graphs

%% then:
%% section vertex leafage: needs a paragraph with definitions...(and more general with tree model?)

\begin{figure}
%\begin{center}
\scalebox{0.84}{
\centering
\hspace*{-0.4in}
\begin{tikzpicture}
\tikzstyle{np}=[fill=lightgray,text=black]
\tikzstyle{p}=[fill=none,text=black]
\tikzstyle{e}=[draw=none]
\matrix (L) [matrix of math nodes, nodes={draw,rounded corners=7pt},row sep=0.6cm,column sep=1.0cm,ampersand replacement=\&]{
\node[e](vlinfty){\begin{tabular}{c}
unbounded\\
vertex leafage
\end{tabular}};\&\&\&\node[np](chordal){\begin{tabular}{c}
$\equiv$ chordal graphs\\
NP-hard~\cite{FominHKPV14}
\end{tabular}};\\
%line 2
\node[e](vl){\begin{tabular}{c}
vertex leafage\\
at most $v\ell \geq 2$
\end{tabular}};\&\&\node[p,text=white](leaf){\begin{tabular}{c}
W[1]-hard, {Theorem~\ref{theo:whard}}
\\[4pt]
\textrm{$O(n^{2\ell+1})$, {Theorem~\ref{theo:leafage}}}
\end{tabular}};
\filldraw[lightgray][] (leaf.west)--(leaf.east)[rounded corners=7pt]--(leaf.north east)[rounded corners=7pt]--(leaf.north west)[rounded corners=7pt]--(leaf.west);
\node[p](leaf){\begin{tabular}{c}
W[1]-hard, {Theorem~\ref{theo:whard}}
\\[4pt]
\textrm{$O(n^{2\ell+1})$, {Theorem~\ref{theo:leafage}}}
\end{tabular}};
\&\node[np](path){\begin{tabular}{c}
$\supseteq$ undirected path graphs\\
NP-hard, {Theorem~\ref{theo:npundirected}}
\end{tabular}};\\
%line 3
\node[e](vl1){\begin{tabular}{c}
vertex leafage \\
at most $1$
\end{tabular}};\&\node[p](interval){\begin{tabular}{c}
$\equiv$ interval graphs\\
$O(nm)$ \cite{PapadopoulosT20}
\end{tabular}};\&\node[p](lvl1){\begin{tabular}{c}
$O(n^{2}m)$
\end{tabular}};\&\node[p](dirpath){\begin{tabular}{c}
$\equiv$ rooted path graphs\\
$O(n^{2}m)$, {Theorem~\ref{theo:directedP}}
\end{tabular}};\\ %\hline
%line 4
\&\node[e](l1){\begin{tabular}{c}
leafage \\
at most $2$
\end{tabular}};\&\node[e](l){\begin{tabular}{c}
leafage\\
at most $\ell \geq 3$
\end{tabular}};\&\node[e](linfty){\begin{tabular}{c}
unbounded\\
leafage
\end{tabular}};\\
};
\draw (-9.0,-1.8) -- (9.5,-1.8);
\draw (-6,-2.9) -- (-6,3.3);
\tikzstyle{every node}=[sloped,fill=white]
\draw[e](interval)--++(0:2.35)node{$\subseteq$}(lvl1)--++(0:2.6)node{$\subseteq$}(dirpath)--node{$\subseteq$}(path)--node{$\subseteq$}(chordal)
(lvl1)--node{$\subseteq$}(leaf)--node{$\subseteq$}(path)
(l1)--++(0:2.35)node{$\subseteq$}(l)--++(0:2.6)node{$\subseteq$}(linfty)
(vl1)--node{$\subseteq$}(vl)--node{$\subseteq$}(vlinfty);
\end{tikzpicture}
}
%\end{center}
\caption{Computational complexity of the SFVS problem parameterized by leafage and vertex leafage.}\label{fig:leafage}
\end{figure}
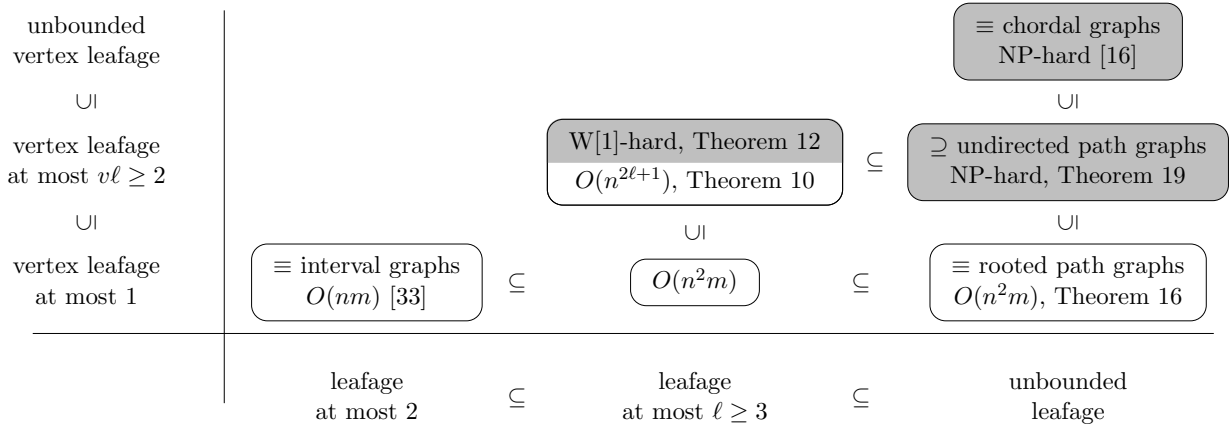

\section{Preliminaries}
All graphs considered here are finite undirected graphs without loops and multiple edges.
We refer to the textbook by Bondy and Murty~\cite{Bondy} for any undefined graph terminology and to the recent book of \cite{CyganFKLMPPS15} for the introduction to Parameterized Complexity.
For a positive integer $p$, we use $[p]$ to denote the set of integers $\{1, \ldots, p\}$.
For a graph $G=(V_G, E_G)$, we use $V_G$ and $E_G$ to denote the set of vertices and edges, respectively.
We use $n$ to denote the number of vertices of a graph and use $m$ for the number of edges.
Given $x\in V_G$, we denote by $N_G(x)$ the neighborhood of $x$.
The \emph{degree} of $x$ is the number of edges incident to $x$.
%%For a set $X\subset V_G$, $N_G(X)$ denotes the set of vertices in $V_G\setminus X$ that have at least one neighbor in $X$.
Given $X\subseteq V_G$, we denote by $G-X$ the graph obtained from $G$ by the removal of the vertices of $X$.
If $X=\{u\}$, we also write $G-u$. The \emph{subgraph induced by $X$} is denoted by $G[X]$, and has $X$ as its vertex set and $\{uv~|~u,v\in X\mbox{ and }uv\in E_G\}$ as its edge set.
A \emph{clique} is a set $K\subseteq V_G$ such that $G[K]$ is a complete graph.

Given a collection $\mathcal{C}$ of sets, the graph $G=(\mathcal{C},\{\{X,Y\}:X,Y \in \mathcal{C} \text{ and }X\cap Y\neq\emptyset\})$ is called \emph{the intersection graph of} $\mathcal{C}$.
Structural properties and recognition algorithms are known for intersection graphs of (directed) paths in (rooted) trees \cite{Chaplick19,Monma86,Panda01}.
%For a graph $G$, any collection $\mathcal{C}$ of sets such that $G$ is its intersection graph is called an \emph{intersection model of} $G$.
Depending on the collection $\mathcal{C}$, we say that a graph is
%A graph is
\begin{itemize}
\item \emph{chordal} if $\mathcal{C}$ is a collection of subtrees of a tree,
\item \emph{undirected path} if $\mathcal{C}$ is a collection of paths of a tree,
\item \emph{rooted path} if $\mathcal{C}$ is a collection of rooted paths of a rooted tree, and
\item \emph{interval} if $\mathcal{C}$ is a collection of subpaths of a path.
\end{itemize}
%% ptolemaic?
%Structural properties and recognition algorithms are known for notions that are closely related to intersection graphs of (directed) paths in (rooted) trees \cite{Chaplick19,Monma86,Panda01}.
For any undirected tree $T$, we use $L(T)$ to denote the set of its leaves, i.e., the set of nodes of $T$ having degree at most one.
If $T$ contains only one node then we let $L(T)=\emptyset$.
Let $T$ be a rooted tree.
We assume that the edges of $T$ are directed towards the root. %, so that every node besides the root has out-degree one.
%For each directed edge $(u, v)$ in $T$, we say that $u$ is a {\em child} of $v$ and that $v$ is a {\em parent} of $u$ (and we denote it by $t(u)$).
%Every vertex except the root has a unique parent.
If there is a (directed) path from node $v$ to node $w$ in $T$, we say that $v$ is a {\em descendant} of $w$ and that $w$ is an {\em ancestor} of $v$. %, and we denote this path by $T[v,w]$.
The leaves of a rooted tree $T$ are exactly the nodes of $T$ having out-degree one and in-degree zero.
Observe that for an undirected tree $T$ with at least one edge we have $|L(T)|\geq 2$, whereas in a rooted tree $T$ with at least one edge $|L(T)|\geq 1$ holds.

A binary relation, denoted by $\leq$, on a set $V$ is called \emph{partial order} if it is transitive and anti-symmetric.
For a partial order $\leq$ on a set $V$, we say that two elements $x$ and $y$ of $V$ are \emph{comparable} if $x \leq y$ or $y \leq x$; otherwise, $x$ and $y$ are called \emph{incomparable}.
If $x \leq y$ and $x \neq y$ then we simply write $x < y$.
Given $X,Y \subseteq V$, we write $X \leq Y$ if for any $x \in X$ and $y \in Y$, we have $x \leq y$;
if $X$ and $Y$ are disjoint then $X \leq Y$ is denoted by $X < Y$.
Given a rooted tree $T$, we define a partial order on the nodes of $T$ as follows: %$u \leq_T v$ if $u$ is an ancestor of $v$.
$x\leq_T y\Leftrightarrow x$ is a descendant of $y$.
It is not difficult to see that if $x \leq_T y$ and $x \leq_T z$ then $y$ and $z$ are comparable, as $T$ is a rooted tree.

\paragraph*{Leafage and vertex leafage} A \emph{tree model} of a graph $G=(V_G,E_G)$ is a pair $(T, \{T_v\}_{v\in V_G})$ where $T$ is a tree, called a \emph{host tree}\footnote{The host tree is also known as a \emph{clique tree}, usually when we are concerned with the maximal cliques of a chordal graph \cite{G74}.},
each $T_v$ is a subtree of $T$, and $uv \in E_G$ if and only if $V(T_u) \cap V(T_v) \neq \emptyset$.
We say that a tree model $(T, \{T_v\}_{v\in V_G})$ \emph{realizes} a graph $H$ if its corresponding graph $G$ is isomorphic to $H$.
It is known that a graph is chordal if and only if it admits a tree model \cite{B74,G74}.
The tree model of a chordal graph is not necessarily unique.
The \emph{leafage} of a chordal graph $G$, denoted by $\ell(G)$, is the minimum number of leaves of the host tree among all tree models that realize $G$, that is,
$\ell(G)$ is the smallest integer $\ell$ such that there exists a tree model $(T, \{T_v\}_{v\in V_G})$ of $G$ with $\ell = |L(T)|$ \cite{LinMW98}.
Moreover, every chordal graph $G$ admits a tree model for which its host tree $T$ has the minimum $|L(T)|$ and $|V(T)| \leq n$ \cite{ChaplickS14,HabibS09};
such a tree model can be constructed in $O(n^3)$ time \cite{HabibS09}.
Thus the leafage $\ell(G)$ of a chordal graph $G$ is computable in polynomial time.

A generalization of leafage is the \emph{vertex leafage} introduced by Chaplick and Stacho \cite{ChaplickS14}.
The vertex leafage of a chordal graph $G$, denoted by $v\ell(G)$, is the smallest integer $k$ such that
there exists a tree mode $(T, \{T_v\}_{v\in V_G})$ of $G$ where $|L(T_v)| \leq k$ for all $v \in V_G$.
Clearly, we have $v\ell(G) \leq \ell(G)$. %, but the converse is not necessarily true.

Although leafage was originally introduced for connected chordal graphs, as opposed to the vertex leafage,
hereafter we relax the connectedness restriction on leafage to avoid confusion between the two notions and
we assume that the considered tree model realizes any chordal graph.
Moreover, we will impose that the host tree $T$ is a rooted tree without affecting structural and algorithmic consequences.
Under these terms, observe that
$\ell(G) = 0$ iff $G$ is a disjoint union of cliques,
$\ell(G) \leq 2$ iff $G$ is an interval graph,
$v\ell(G) \leq 1$ iff $G$ is a rooted path graph, and
$v\ell(G) \leq 2$ iff $G$ is an undirected path graph.

\paragraph*{$S$-forests and $S$-triangles}
By an induced cycle of $G$ we mean a chordless cycle. %.
A \emph{triangle} is a cycle on $3$ vertices.
%A cycle on three vertices is referred to as a \emph{triangle}.
Hereafter, we consider subclasses of chordal graphs, that is graphs that do not contain induced cycles on more than $3$ vertices.

Given a graph $G$ and $S\subseteq V(G)$, we say that a cycle of $G$ is an \emph{$S$-cycle} if it contains a vertex in $S$.
Moreover, we say that an induced subgraph $F$ of $G$ is an \emph{$S$-forest} if $F$ does not contain an $S$-cycle.
%As we consider chordal graphs, it is not difficult to see that
Thus an induced subgraph $F$ of a chordal graph is an $S$-forest if and only if $F$ does not contain any $S$-triangle.
%\todo{maybe state it as an observation?}
%\todo[color=blue!40]{Too obvious. Induced $S$-triangle less so, but still\ldots}
%\todo{Agree, but let's try to stick with $S$-triangles instead of $S$-cycles, all over the text. }
Typically, the \textsc{Subset Feedback Vertex Set} problem asks for a vertex set of minimum (weight) size such that its removal results in an $S$-forest.
The set of vertices that do not belong to an $S$-forest is referred to as \emph{subset feedback vertex set}.
In our dynamic programming algorithms, we focus on the equivalent formulation of computing a maximum weighted $S$-forest.

%Our algorithms compute a maximum $S$-forest of $G$ via a dynamic programming scheme.
For a collection $\mathcal{C}$ of sets, we write $\displaystyle\max_{\text{weight}}\{C\in\mathcal{C}\}$ to denote $\arg\max_{C\in\mathcal{C}}\{weight(C)\}$,
where $weight(C)$ is the sum of weights of the vertices in $C$.
The collection of $S$-forests of a graph $G$, is denoted by $\mathcal{F}_{S}$.
%
%\begin{definition}
Let $X,Y\subseteq V_{G}$ such that $X\cap Y=\emptyset$ and $G[Y]\in\mathcal{F}_{S}$. Then,
$\displaystyle\optdir{X}{Y}=\max_{\text{weight}}\{U\subseteq X:G[U\cup Y]\in\mathcal{F}_{S}\}$.
%\begin{displaymath}
%\optdir{X}{Y}=\max_{w}\{U\subseteq X:G[U\cup Y]\in\mathcal{F}_{S}\}.
%\end{displaymath}
%\end{definition}
%
%\vspace*{-0.05in}
\begin{itemize}
\item
Our desired optimal solution is $\displaystyle A_{V_{G}}^{\emptyset} = \max_{\text{weight}}\{U\subseteq V_{G}:G[U]\in\mathcal{F}_{S}\}$.
We will subsequently show that in order to compute $A_{V_{G}}^{\emptyset}$ it is sufficient to compute $\optdir{X}{Y}$ for
a polynomial number of sets $X$ and $Y$.
\end{itemize}
%Observe that $A_{\emptyset}^{Y}=\emptyset$ for any $Y$.

Let $G=(V_G, E_G)$ be a chordal graph and let $X,Y\subseteq V_{G}$ such that $X\cap Y=\emptyset$ and $G[Y]\in\mathcal{F}_{S}$.
A partition $\mathcal{P}$ of $X$ is called \emph{nice} if for any $S$-triangle $S_t$ of $G[X \cup Y]$,
there is a partition class $P_i \in \mathcal{P}$ such that $V(S_t) \cap X \subseteq P_i$.
In other words, any $S$-triangle of $G[X \cup Y]$ is involved with at most one partition class of a nice partition $\mathcal{P}$ of $X$.
%\todo[inline,color=red!40]{As written, this says `there is a class that for every $S$-cycle it contains it', whereas it should say `for every $S$-cycle there is a class that contains it'\ldots}
%
With respect to the optimal defined solutions $\optdir{X}{Y}$, we observe the following:
%\footnote{In this extended abstract, proofs of statements marked with an asterisk (*) were removed to an appendix.}
%\footnote{Due to space constraints, all proofs can be found in the appendix.}:
%\begin{restatable}[*]{observation}{obsAXY}\label{obs:AXY}
%\todo[inline,color=red!40]{In the following Obs, you hid the first obs in the proof as if it is only a stepping stone for the second one. However, the two obs's are independent and they are both necessary for formulas like the one in Lemma \ref{lem:together} as the first one reduces $Y$ and the second one partitions $X$.}
\begin{observation}\label{obs:AXY}
Let $G=(V_G, E_G)$ be a chordal graph and let $X,Y\subseteq V_{G}$ such that $X\cap Y=\emptyset$ and $G[Y]\in\mathcal{F}_{S}$. 
Then, the following hold:
%For any nice partition $\mathcal{P}$ of $X$, we have $\optdir{X}{Y}=\bigcup_{X'\in\mathcal{P}}{\optdir{X'}{Y}}$.
\renewcommand\labelenumi{(\theenumi)}\begin{enumerate}
\item\label{obs:XY:between} $\optdir{X}{Y}=\optdir{X}{Y'}$ for any $Y\supseteq Y'\supseteq Y\cap N(X')$ where $X'=X\setminus\{u\in X\setminus S:Y\cap N(u)\subseteq Y\setminus S\}$.
\item\label{obs:XY:nice} $\optdir{X}{Y}=\bigcup_{X'\in\mathcal{P}}{\optdir{X'}{Y}}$ for any nice partition $\mathcal{P}$ of $X$.
\end{enumerate}
\end{observation}
\begin{proof}
For the first statement, observe that $G[Y'] \in \mathcal{F}_{S}$, as an induced subgraph of an $S$-forest.
Also, notice that any $S$-triangle in $G[X \cup Y']$ remains an $S$-triangle in $G[X \cup Y]$. 
Consider an $S$-triangle $\{a,b,x\}$ in $G[X \cup Y]$ with $x \in X$ and $a \in Y$.
We show that $a \in Y'$ and $b \in X \cup Y'$.
%If $x \in S$ then $x \in X'$ which means that $a \in Y'$ and $b \notin Y \setminus Y'$ by the fact that $N(X') \subseteq Y'$.
%Similarly, if $x \in X'$ then we conclude the claim.
%<S> All together now! </S>
If $x \in X'$, then $a \in Y'$ and $b \in X \cup Y'$ by the fact that $Y\cap N(X')\subseteq Y'$.
Suppose that $x \in X \setminus S$ such that $Y\cap N(x)\subseteq Y\setminus S$.
This means that the only vertex of $S$ in the $S$-triangle is $b$.
In particular, we have $b \in X \cap S$ and, since $a \in N(b)$, we conclude that $a \in Y'$.
Thus, any $S$-triangle in $G[X \cup Y]$ remains an $S$-triangle in $G[X \cup Y']$, which shows the claim.
% Notice than if a vertex $u\in X\setminus S$ such that $Y\cap N(u)\subseteq Y\setminus S$ and a vertex $w\in Y$ participate in
% the formation of an $S$-triangle of $G[X\cup Y]$, then the remaining participating vertex must be some vertex $v\in X\cap S$, so $w$ is in $N(X')$ as a neighbour of that vertex $v$.

For the second statement, assume that there is an $S$-triangle $S_t$ in $G[X \cup Y]$.
Then it must contain a vertex $v$ of some partition class $P_i$, as $G[Y]$ is an $S$-forest.
By the definition of a nice partition $\mathcal{P}$, we have $V(S_t)\cap X\subseteq P_i$.
%In particular, for any partition class $P_i \in \mathcal{P}$ and any $S$-triangle $S_t$, we have either $V(S_t)\textcolor{green}{\cap X}\subseteq P_i$ or $V(S_t) \cap P_i = \emptyset$.
Therefore, we deduce $\optdir{X}{Y}\cap P_{i}=\optdir{P_{i}}{Y}$, which concludes the proof.
\end{proof}

By Observation~\ref{obs:AXY}, we search for nice partitions of the vertex set $X$ in order to consider smaller instances of $\optdir{X}{Y}$.
More precisely, 
Observation~\ref{obs:AXY}~(ii) suggests how to consider the natural sets $X'$ of a nice partition of $X$, 
whereas Observation~\ref{obs:AXY}~(i) indicates which vertices of $Y$ are relative to each set $X'$.   

\section{Expanded tree model and related vertex subsets}\label{sec:expanded}
Given a tree model of a chordal graph, we are interested in defining a partial order on the vertices of the graph that takes advantage the underlying tree structure.
For this reason, it is more convenient to consider the tree model as a natural rooted tree and each of its subtrees to correspond to at most one maximal vertex.
Here we show how a tree model can be altered in order to capture the appropriate properties in a formal way.
We assume that $G$ is a chordal graph that admits a tree model $(T, \{T_v\}_{v\in V_G})$ such that $|L(T)| = \ell(G)$.
We will concentrate on the case in which $|L(T)| \geq 2$ and $T$ contains a non-leaf node.
The rest of the cases (i.e., $|V(T)| \leq 2$) are handled by the algorithm on interval graphs~\cite{PapT19} in a separate way.
For this purpose we say that a chordal graph $G$ is \emph{non-trivial} if $|V(T)| > 2$.

\begin{definition}
A tree model $(T, \{T_v\}_{v\in V_G})$ of $G$ is called \emph{expanded tree model} if
\begin{itemize}
\item the host tree $T$ is rooted (and, consequently, all of its subtrees are rooted),
\item for every $v\in V_G$, $L(T_v) \neq \emptyset$ holds, and
\item every node of $T$ is either the root or a leaf of at most one subtree $T_v$ that corresponds to a vertex $v$ of $G$.
\end{itemize}
\end{definition}
%\todo[inline]{I am not sure if the above definition, allows a node of $T$ to be at the same time a root of a vertex and a leaf of another one. In my understanding, I don't see a problem to allow such a case, as we mainly visit the vertices according to the partial order of their roots. And to be more concrete: why do we care if a node of the tree is a leaf in several subtrees?}

%\todo[inline,color=blue!40]{I do not like the term `extended' here, we do not add any new information to it, maybe `expanded' (as in blown-up) is more appropriate\ldots}

We show that any non-trivial chordal graph admits an expanded tree model that is \emph{close} to its tree model.
In fact, we provide an algorithm that, given a tree model of a non-trivial chordal graph $G$, constructs an expanded tree model that realizes $G$.

%\begin{restatable}{lemma}{lemexpanded}\label{lem:expanded}
\begin{lemma}\label{lem:expanded}
For any tree model $(T, \{T_v\}_{v\in V_G})$ of $G$ with $|L(T)|=\ell\geq2$ and $|L(T_v)|\leq v\ell \leq \ell$ for all $v\in V_G$, %$\ell = |L(T)|$ and $3 \leq |V(T)|\leq n$,
there is an expanded tree model $(T', \{T'_v\}_{v\in V_G})$ of $G$ such that:
\begin{itemize}
\item $|L(T')| = \ell$,
\item $|L(T_v)|-1\leq|L(T'_v)|\leq|L(T_v)|$ for every $v \in V_G$, and
\item $|V(T')|\leq|V(T)|+(1+v\ell)(n-1)$.
\end{itemize}
Moreover, given $(T, \{T_v\}_{v\in V_G})$, the expanded tree model can be constructed in time $O(n^2)$.
\end{lemma}
%\end{restatable}
\begin{proof}
We root $T$ at a non-leaf node of $T$, resulting in a rooted tree $T'$ with $|L(T')| = \ell$.
Moreover, we root every $T_v$ at the node of $T_{v}$ which is closer to $r(T')$, resulting in a rooted subtree $T_{v}'$.
Notice that $|L(T_{v})|-1\leq|L(T_{v}')|\leq|L(T_{v})|$, as $r(T_{v}')$ may be a leaf of $T_{v}$.
In what follows, we assume that $T$ and all of its subtrees $\{T_v\}_{v\in V_G}$ are rooted trees.

\begin{figure}[th]
\scalebox{1.0}{
\centering
\noindent
\begin{tikzpicture}
\tikzstyle{b}=[circle,inner sep=2pt,draw,fill]
\tikzstyle{w}=[circle,inner sep=2pt,draw,fill=white]
\tikzstyle{t}=[circle,inner sep=2pt,draw,fill=white,dotted]
\tikzstyle{p}=[decorate,decoration={snake,amplitude=.4mm,segment length=2mm,post length=1mm}]
%before
\draw(0,0)node[b](x){}node[anchor=south west]{$x$}+(120:1)node[w](y1){}+(150:1)node[w](y2){}+(180:1)node[w](y3){}+(210:1)node[t](y4){}+(240:1)node[w](y5){}+(0:1)node[w](y0){};
\draw[->](y1)--(x);
\draw[->](y2)--(x);
\draw[->](y3)node[anchor=east]{$N^{-}(x)$}--(x);
\draw[->,dotted](y4)--(x);
\draw[->](y5)--(x);
\draw[->](x)--(y0)node[anchor=west]{$N^{+}(x)$};
%to
\node at(4,0){$\longmapsto$};
%after
\draw(8,0)node[b](xminus){}node[anchor=south west]{$x_{-k_{l}}$}+(120:1)node[w](y1){}+(150:1)node[w](y2){}+(180:1)node[w](y3){}+(210:1)node[t](y4){}+(240:1)node[w](y5){}++(0:1.5)node[b](x0){}node[anchor=south west]{$x_{0}$}++(0:1.5)node[b](xplus){}node[anchor=south west]{$x_{k_{r}}$}+(0:1)node[w](y0){};
\draw[->](y1)--(xminus);
\draw[->](y2)--(xminus);
\draw[->](y3)node[anchor=east]{$N^{-}(x)$}--(xminus);
\draw[->,dotted](y4)--(xminus);
\draw[->](y5)--(xminus);
\draw[->,p](xminus)--(x0);
\draw[->,p](x0)--(xplus);
\draw[->](xplus)--(y0)node[anchor=west]{$N^{+}(x)$};
\end{tikzpicture}
}
\caption{We replace node $x$ of $T$ by the directed path $\langle x_{-k_{l}},\ldots,x_{0},\ldots,x_{k_{r}}\rangle$ such that the nodes $N^{-}(x)$ now point to $x_{-k_{l}}$ and node $N^{+}(x)$ is now pointed by $x_{k_{r}}$.}\label{fig:expandtree}
%the nodes that point to $x$ now point to $x_{-k_{l}}$ and the node that is pointed by $x$ is now pointed by $x_{k_{r}}$.}\label{fig:expandtree}
\end{figure}
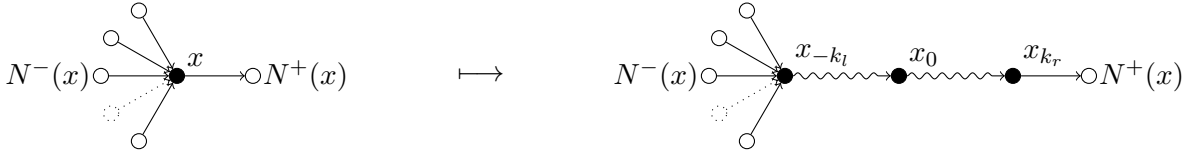

Consider a node $x$ of $T$. Assume that $x$ is the root of $k_{r}$ subtrees $T_{v_{1}},\ldots,T_{v_{k_{r}}}$ and a leaf of $k_{l}$ subtrees $T_{v_{-1}},\ldots,T_{v_{-k_{l}}}$ of $\{T_v\}_{v\in V_G}$ where $k_{r}+k_{l}\geq2$. In this context, for every $v\in V_{G}$ corresponding to $T_v = \{x\}$, we consider $x$ as being both the root and a leaf of $T_{v}$. We replace the node $x$ in $T$ by the gadget shown in Figure~\ref{fig:expandtree}. We also modify every subtree $T_{v}$ of $\{T_v\}_{v\in V_G}$ as follows:
\begin{itemize}
\item If $T_{v}=T_{v_{i}}$ for some $i\in\{1,\ldots,k_{r}\}$ and $T_{v}\neq T_{v_{j}}$ for all $j\in\{-1,\ldots,-k_{l}\}$, then we replace $x$ in $T_{v}$ by the part of the gadget involving the vertices $x_{-k_{l}},\ldots,x_{0},\ldots,x_{i}$.
\item If $T_{v}\neq T_{v_{i}}$ for all $i\in\{1,\ldots,k_{r}\}$ and $T_{v}=T_{v_{j}}$ for some $j\in\{-1,\ldots,-k_{l}\}$, then we replace $x$ in $T_{v}$ by the part of the gadget involving the vertices $x_{j},\ldots,x_{0},\ldots,x_{k_{r}}$.
\item If $T_{v}=T_{v_{i}}$ for some $i\in\{1,\ldots,k_{r}\}$ and $T_{v}=T_{v_{j}}$ for some $j\in\{-1,\ldots,-k_{l}\}$, then we replace $x$ in $T_{v}$ by the part of the gadget involving the vertices $x_{j},\ldots,x_{0},\ldots,x_{i}$.
\item If $T_{v}\neq T_{v_{i}}$ for all $i\in\{1,\ldots,k_{r}\}$ and $T_{v}\neq T_{v_{j}}$ for all $j\in\{-1,\ldots,-k_{l}\}$, then $T_{v}'=T_{v}$.
\end{itemize}
To see that new model indeed realizes $G$, observe that for every $T_{u},T_{w}\in \{T_v\}_{v\in V_G}$:
\begin{itemize}
\item if $x\in V(T_{u})\cap V(T_{w})$, then $x_{0}\in V(T_{u}')\cap V(T_{w}')$, and
\item if $x\notin V(T_{u})\cap V(T_{w})$, then $x_{-k_{l}},\ldots,x_{k_{r}}\notin V(T_{u}')\cap V(T_{w}')$.
\end{itemize}
Thus the intersection graph of $(T', \{T'_v\}_{v\in V_G})$ is isomorphic to $G$.
Notice that any node $x_{i}$, $i\in\{-k_{l},\ldots,k_{r}\}$ is either the root or a leaf of at most one subtree of $\mathcal{T}'$ and,
in particular, $|L(T_{v}')|=1$ for any $T_{v}=\{x\}$.
Iteratively applying the above modifications to $T$ and $\{T_v\}_{v\in V_G}$, results in an expanded tree model $\mathcal{T}'$ of $G$ that satisfies the claimed properties. %the first three properties stated in this Lemma hold.

To bound $|V(T')|$, observe that the first step adds at most $n$ new nodes in $T$, so that $|V(T')| \leq |V(T)| + n$.
Further notice that every subtree $T_v$ has at most $v\ell$ leaves.
%%%This means than among the subtrees $\{T_v\}_{v\in V_G}$ the number of common roots is at most $n$ and the number of common leaves is at most $\ell n$.
%%%Since we add exactly $k$ new nodes in $T$ for $k$ common roots or leaves and $|V(T)|\leq n$, we conclude $|V(T')|\leq (\ell+3)n$.
In the worst case, all subtrees of $\{T_v\}_{v\in V_G}$ are rooted in the same node and all their leaves are contained in a set of $\max_{v\in V_{G}}|L(T_{v})|=v\ell$ nodes, so our preprocessing algorithm will add $(1+v\ell)(n-1)$ nodes to $T$.
Moreover, as we need to update $n+1$ trees by adding at most a total of $(1+v\ell)(n-1)$ new nodes and $|V(T)| \leq n$, the total running time is $O(n^2)$.
\end{proof}

%\paragraph{Vertex subsets on expanded tree model}
%\subsection{Vertex subsets on expanded tree model}
Hereafter we assume that $(T, \{T_v\}_{v\in V_G})$ is an expanded tree model of a non-trivial chordal graph $G$.
For any vertex $u$ of $G$, we denote the root of its corresponding rooted tree $T_u$ in $T$ by $r(u)$.
We define the following partial order on the vertices of $G$: for all $u,v\in V_{G}$, $u\leq v\Leftrightarrow r(u)\leq_{T} r(v)$.
%\begin{itemize}
%\item for all $u,v\in V_{G}$, $u\leq v\Leftrightarrow r(u)\leq_{T} r(v)$.
%\end{itemize}
%
\noindent In other words, two vertices of $G$ are comparable (with respect to $\leq$) if and only if there is a directed path between their corresponding roots in $T$.
For all $u\in V_{G}$, we define $V_{u}=\{u'\in V_{G}:u'\leq u\}$. %% and $V_{u\to u}\equiv V_{u}$.
%% explain what it means...
%In terms of the host tree, this set corresponds to all vertex-descendants of $r(u)$ that are entirely contained within the subtree rooted at $r(u)$.

%\begin{restatable}[*]{observation}{obsordering}\label{obs:ordering}
\begin{observation}\label{obs:ordering}
Let $u,v,w,z \in V_{G}$. Then, the following hold:
\renewcommand\labelenumi{(\theenumi)}\begin{enumerate}
%\item The vertices of $V_{u}\cap N(u)$ are pairwise comparable.
%\todo[inline]{I don't see that for rooted subtrees...; but perhaps we don't need it for subtrees, rather for paths. }
%\todo[color=red!40,inline]{Indeed, this does not hold for subtrees in general. A vertex corresponding to a path induces a total order to its neighbourhood solely because all its neighbours' roots are contained in that single path. However, I cannot recall or find this being explicitly used anywhere\ldots}
\item \label{obs:ord:comp} If $uv \in E_{G}$, then $u$ and $v$ are comparable.
\item \label{obs:ord:twocomp} If $u \leq v$, $w\leq z$, and $u$ and $w$ are comparable, then $v$ and $z$ are comparable.
%For all $u_{1},u_{2},v_{1},v_{2}\in V_{G}$ such that $u_{1}\leq v_{1}$ and $u_{2}\leq v_{2}$, if $u_{1}$ and $u_{2}$ are comparable, then $v_{1}$ and $v_{2}$ are comparable.
\item \label{obs:ord:umbrella} If $u<v<w$ and $uw\in E_{G}$, then $vw\in E_{G}$.
\end{enumerate}
\end{observation}
\begin{proof}
%Assume that $u$ and $v$ are incomparable. Since there is no directed path between $r(u)$ and $r(v)$, the subtrees of $T$ rooted at $r(u)$ and $r(v)$ are disjoint. Thus $T_u$ and $T_v$ are disjoint in $T$, as they are contained in the subtrees rooted at $r(u)$ and $r(v)$, respectively. This means that $uv \notin E_G$, so that $u$ and $v$ are comparable whenever $uv \in E_G$.
%\todo[color=red!40,inline]{The proof above simply claims that the hypothesis (there is a directed path between the roots) implies the result (the corresponding subtrees are disjoint) without ever showing why.}
%\todo[color=green!40,inline]{Replaced.}
Assume that $x\in V(T_{u})\cap V(T_{v})$, which exists as $uv \in E_{G}$. Then there are paths $x\to r(u)$ and $x\to r(v)$. Since $T$ is a rooted tree, any node besides its root has a unique parent. This implies that the shortest of the aforementioned paths is a subpath of the longest. Assume, without loss of generality, that $x\to r(u)$ is the shortest path. Then $x\to r(v)=x\to r(u)\to r(v)$, so that $u$ and $v$ are indeed comparable.

For the second statement, observe that all ancestors of a node of $T$ are pairwise comparable with respect to $\leq_T$.
Assume that $u\leq w$. Then $r(u)\leq r(w)\leq r(z)$ because $w\leq z$, so $u\leq z$ in addition to $u\leq v$.
Now assume that $w\leq u$. Then $r(w)\leq r(u)\leq r(v)$ because $u\leq v$, so $w\leq v$ in addition to $w\leq v$.
In both cases we conclude that $v$ and $z$ are comparable.

For the third statement, observe that $u<v<w$ implies that $r(u)<r(v)<r(w)$ which in turn implies that $r(u)\to r(v)\to r(w)$.
We show that $r(u)\in V(T_{w})$. Since $u$ and $w$ are adjacent, there exists a node $x\in V(T_{u})\cap V(T_{w})$. As shown in the proof of the first statement, there exists a path $x\to r(u)\to r(w)$.
Then all the nodes of this path are in $V(T_{w})$ because its endpoints are in $V(T_{w})$ and $T_{w}$ is connected.
Thus we have $r(u)\in V(T_{w})$.
With the same argumentation, we conclude that all nodes of the path $r(u)\to r(v)\to r(w)$ are in $V(T_{w})$, so that $r(v)\in V(T_{w})$ holds.
Therefore, $v$ and $w$ are adjacent.
\end{proof}

%\end{restatable}

%\begin{restatable}{lemma}{propneighb}\label{prop:neighb}
\begin{lemma}\label{prop:neighb}
For every $u\in V_{G}$, we have $N(V_{u})\subseteq N(u)$.
\end{lemma}
%\end{restatable}
\begin{proof}
Let $w\in N(V_{u})$. Then there is a vertex $v\in V_{u}$ such that $vw\in E_{G}$.
It suffices to show that $w \in N(u)$. Assume that $u\neq v$, as otherwise the claim trivially holds.
Then $v < u$, because $v\in V_{u}$. Moreover, Observation~\ref{obs:ordering}~({\ref{obs:ord:comp}}) implies that $w<v$ or $v<w$.
Since $w\notin V_{u}$, we conclude $u < v < w$.
Therefore Observation \ref{obs:ordering}~(\ref{obs:ord:umbrella}) shows that $uw\in E_{G}$.
%by Observation \ref{obs:ordering}~(\ref{obs:ord:umbrella}).
\end{proof}

For all $u\in V_{G}$, we denote the set of all maximal proper predecessors of $u$ by $\vartriangleleft u$.
Notice that such vertices correspond to the maximal descendants of $r(u)$.
%With respect to $T$, this set contains the highest vertex-descendants for which their vertex-subtree is entirely contained within the subtree rooted at $r(u)$.
%\todo[color=red!40,noline]{I see no purpose for this remark.}
For all $U\subseteq V_{G}$, we define $\mathcal{V}_{U}=\{V_{u}:u\in U\}$.
We extend the previous case of a single vertex, on subsets of vertices with respect to an edge.
For all $u,v\in V_{G}$ such that $uv\in E_{G}$, we denote by $\vartriangleleft uv$ the set of all maximal vertices of $V_{G}$ that are proper predecessors of both $u$ and $v$ but are not adjacent to both, so $\vartriangleleft uv=\max_{G}((V_{u}\cap V_{v})\setminus(N[u]\cap N[v]))$.
Recall that for any edge $uv \in E_{G}$, either $u < v$ or $v < u$ by Observation~\ref{obs:ordering}~(\ref{obs:ord:comp}).
If $u<v$ holds, then $\vartriangleleft uv=\max_{G}(V_{u}\setminus(N[u]\cap N(v)))$.
The following two lemmas are crucial for our algorithms, as they provide natural partitions into smaller instances.

%\begin{restatable}[*]{lemma}{propdisc}\label{prop:disc}
\begin{lemma}\label{prop:disc}
For every $u\in V_{G}$, the collection $\mathcal{V}_{\vartriangleleft u}$ is a partition of $V_{u}\setminus\{u\}$ into pairwise disconnected sets.
For every $u,v\in V_{G}$ such that $u<v$ and $uv\in E_{G}$, $\mathcal{V}_{\vartriangleleft uv}$ is a partition of $V_{u}\setminus(N[u]\cap N(v))$ into pairwise disconnected sets.
\end{lemma}
\begin{proof}
We prove the first statement. The proof of the second statement is completely analogous.
Firstly notice that, by definition, the vertices of $\vartriangleleft u$ are pairwise incomparable. Consider two vertices $u_{1}'$ and $u_{2}'$ such that $u_{1}'\leq u_{1}$ and $u_{2}'\leq u_{2}$ where $u_{1}$ and $u_{2}$ are two vertices of $\vartriangleleft u$.
Clearly, $u_{1}' \in V_{u_{1}}$ and $u_{2}' \in V_{u_{2}}$.
By Observation~\ref{obs:ordering} (\ref{obs:ord:comp}--\ref{obs:ord:twocomp}), it follows that $u_{1}'$ and $u_{2}'$ are distinct and non-adjacent.
\end{proof}

%\end{restatable}

%\begin{restatable}[*]{lemma}{lemniceofundirected}\label{lem:niceofundirected}
\begin{lemma}\label{lem:niceofundirected}
For every $u\in V_{G}$, the collection $\mathcal{V}_{\vartriangleleft u}$ is a nice partition of $V_{u}\setminus\{u\}$.
For every $u,v\in V_{G}$ such that $u<v$ and $uv\in E_{G}$, the collection $\mathcal{V}_{\vartriangleleft uv}$ is a nice partition of $V_{u}\setminus(N[u]\cap N(v))$.
\end{lemma}
\begin{proof}
We prove the first statement. The proof of the second statement is completely analogous. Let $X=V_{u}\setminus\{u\}$ and $Y\subseteq V_{G}$ such that $X\cap Y=\emptyset$. Suppose that $S_{t}$ is an $S$-triangle of $G[X\cup Y]$ for which the intersection of $V(S_{t})$ and a class of $\mathcal{V}_{\vartriangleleft u}$ is non-empty for at least two such classes. Assume that $P_{1}$ and $P_{2}$ are two of those classes and let $u_{1}\in V(S_{t})\cap P_{1}$ and $u_{2}\in V(S_{t})\cap P_{2}$. Then $u_{1}$ and $u_{2}$ must be adjacent, which is in contradiction to Lemma~\ref{prop:disc}.
\end{proof}

%\end{restatable}

Having defined the necessary predecessors (maximal descendants) of $u$, we next analyze specific solutions described in $\optdir{V_u}{Y}$ with respect to the vertices of $\vartriangleleft u$.
Both statements follow by carefully applying Lemma~\ref{prop:neighb} and Lemma~\ref{lem:niceofundirected}.
%More precisely, we deal first with the case in which $u \notin \optdir{V_{u}}{Y}$.

%\begin{restatable}[*]{lemma}{lemtogether}\label{lem:together}
\begin{lemma}\label{lem:together}
Let $Y \subseteq V_{G} \setminus V_u$. \textup{(i)}\ If $u \notin \optdir{V_{u}}{Y}$ then $\displaystyle\optdir{V_{u}}{Y}=\bigcup_{u'\in\vartriangleleft u}{\optdir{V_{u'}}{Y\cap N(u')}}$.\\
\textup{(ii)}\ Moreover,
$\displaystyle\optdir{V_{u}}{\emptyset}=\max_{\text{weight}}\left\{\bigcup_{u'\in\vartriangleleft u}{\optdir{V_{u'}}{\emptyset}},\{u\}\cup\bigcup_{u'\in\vartriangleleft u}{\optdir{V_{u'}}{\{u\}\cap N(u')}}\right\}$.
\end{lemma}
\begin{proof}
We first show claim~(i). Since $u\notin\optdir{V_{u}}{Y}$, we have $\optdir{V_{u}}{Y}=\optdir{V_{u}\setminus\{u\}}{Y}$.
According to Lemma~\ref{lem:niceofundirected}, the collection $\mathcal{V}_{\vartriangleleft u}$ is a nice partition of $V_{u}\setminus\{u\}$.
%, as there is no edge among the vertices of pairwise elements of $\mathcal{V}_{\vartriangleleft u}$.
Thus, by Observation~\ref{obs:AXY} and Lemma~\ref{prop:neighb}, we have
\begin{displaymath}
\optdir{V_{u}\setminus\{u\}}{Y}=\bigcup_{X\in \mathcal{V}_{\vartriangleleft u}}{\optdir{X}{Y}}=\bigcup_{u'\in\vartriangleleft u}{\optdir{V_{u'}}{Y\cap N(u')}}.
\end{displaymath}

For the second claim, we distinguish two cases depending on whether $u$ is in $\optdir{V_{u}}{\emptyset}$ or not.
If $u\notin\optdir{V_{u}}{\emptyset}$ then claim~(i) shows the described formula.

Assume that $u\in\optdir{V_{u}}{\emptyset}$. Then $\optdir{V_{u}}{\emptyset}=\{u\}\cup\optdir{V_{u}\setminus\{u\}}{\{u\}}$.
Recall that the collection $\mathcal{V}_{\vartriangleleft u}$ is a nice partition of $V_{u}\setminus\{u\}$.
Moreover, if $u$ has no neighbor in $V_{u'}$ then $u \notin N(u')$ by Lemma~\ref{prop:neighb}. Thus, we get the desired formula:
%According to Proposition ? 1--2 the collection $\mathcal{V}_{\vartriangleleft u}$ again has the desired property so that
\begin{displaymath}
\optdir{V_{u}\setminus\{u\}}{\{u\}}=\bigcup_{X\in \mathcal{V}_{\vartriangleleft u}}{\optdir{X}{\{u\}}}=\bigcup_{u'\in\vartriangleleft u}{\optdir{V_{u'}}{\{u\}\cap N(u')}}.
\end{displaymath}
\end{proof}

\section{SFVS on graphs with bounded leafage}\label{sec:leaf}
%\subsection{SFVS on graphs with bounded leafage}
% definition of bounded leafage graph.
In this section we concern ourselves with chordal graphs that have an intersection model tree with at most $\ell$ leaves.
Our goal is to show that SFVS can be solved in polynomial time on chordal graphs with bounded leafage.
In particular, we show that that SFVS is in XP parameterized by $\ell$.
%% by providing an algorithm with running time $O((\ell+1)n^{2\ell+1})$.
%In the case of $\ell=1$, observe that the input graph is a complete graph, so SFVS can be solved in $O(n^{2})$ by comparing all maximal subset feedback vertex sets because a clique of an $S$-forest containing a vertex of $S$ must contain at most two vertices.
In the case of $\ell \leq 2$, the input graph is an interval graph, so SFVS can be solved in $O(nm)$ time \cite{PapT19}.
We subsequently assume that we are given a chordal graph $G$ that admits an expanded tree model $(T, \{T_v\}_{v\in V_G})$ with $\ell=L(T)\geq2$, due to Lemma~\ref{lem:expanded}.
%The extreme case of $G$ being a trivial chordal graph (i.e., $|V(T)| \leq 2$) will be handled separately.

Given a subset of vertices of $G$, we collect the leaves of their corresponding subtrees: for every $U\subseteq V_{G}$, we define $L(U)=\cup_{u\in U}{L(T_{u})}$.
Notice that for any non-empty $U\subseteq V_{G}$, we have $L(U) \neq \emptyset$, since $(T, \{T_v\}_{v\in V_G})$ is an expanded tree model.
Moreover, we associate the nodes of $T$ with the vertices of $G$ for which the nodes appear as leaves in their corresponding subtrees:
for every $V\subseteq V_{T}$, we define $L^{-1}(V)$ to be the set $\{u\in V_{G}:L(T_{u})\cap V\neq\emptyset\}$.
For $V\subseteq V_{T}$, we denote by $\min_T V$ the subset of minimal nodes of $V$ with respect to $\leq_T$.
Observe that $\min_T V$ is a set of pairwise incomparable nodes, so $|\min_{T}V|\leq |\min_T V_{T}| \leq \ell$.%, as $T$ has at most $\ell$ leaves.
%% perhaps, we could go for = intead of \leq but that would require $L(T)=\ell$.

%\todo[color=green!40,inline]{Major changes in the following lemma and subsequent proof.}
%\begin{restatable}{lemma}{lemLL}\label{lem:LL-1}
\begin{lemma}\label{lem:LL-1}
Let $U\subseteq V_{G}$ and $V\subseteq L(U)$. Then $L^{-1}(V)\subseteq U$.
\end{lemma}
%\end{restatable}
\begin{proof}
The fact that $V\subseteq L(U)$ yields $L^{-1}(V)\subseteq L^{-1}(L(U))$.
We will show that $L^{-1}(L(U))\subseteq U$.
Let $u$ be a vertex of $G$ such that $u\notin U$. Then $L(T_{u})\cap L(U)=\emptyset$, because $(T, \{T_v\}_{v\in V_G})$ is an expanded tree model. Thus $u\notin L^{-1}(L(U))$.
\end{proof}

Instead of manipulating with the actual vertices of $U$, our algorithm deals with the \emph{representatives} of $U$ which contain the vertices of $L^{-1}(\min_T L(U))$.
In particular, we are interested in the set of vertices $F_{\leq2}(U)=F_{1}(U)\cup F_{2}(U)$, where $F_{1}(U)=L^{-1}(\min_T\{L(U)\})$ and $F_{2}(U)=L^{-1}(\min_T\{L(U\setminus F_{1}(U))\})$.
%\begin{itemize}
%\item $F_{1}(U)=L^{-1}(\min_T\{L(U)\})$ \quad and \quad $F_{2}(U)=L^{-1}(\min_T\{L(U\setminus F_{1}(U))\})$.
%\item $F_{2}(U)=L^{-1}(\min_T\{L(U\setminus F_{1}(U))\})$.
%\end{itemize}
%which are exactly the vertices of $U_1 \cup U_2$, where $U_1 = L^{-1}(\min_T L(U))$ and $U_2 = L^{-1}(\min_T L(U \setminus U_1))$.
We show that the representatives hold all the necessary information needed from their actual vertices.
%% Note that in each case, the left-sided formula is obtained from Lemma...

%\begin{restatable}[*]{lemma}{lemAleaf}\label{lem:Aleaf}
\begin{lemma}\label{lem:Aleaf}
Let $u\in V_{G}$ and $W\subseteq V_{G}\setminus V_{u}$ such that $W\neq\emptyset$, $G[\{u\}\cup W]$ is a clique, and $G[W]\in\mathcal{F}_{S}$,
and let $u\in\optdir{V_{u}}{W}$.
\begin{itemize}
\item If $(\{u\}\cup W)\cap S \neq\emptyset$ then $W=\{w\}$ and no vertex of $V_{u}\cap N(u)\cap N(w)$ belongs to $\optdir{V_{u}}{\{w\}}$.
\item If $(\{u\}\cup W)\cap S=\emptyset$ then $\optdir{V_{u'}}{W\cap N(u')}=\optdir{V_{u'}}{F_{\leq2}((\{u\}\cup W)\cap N(u'))}$, for any vertex $u'\in\vartriangleleft u$.
\end{itemize}
%% $W_S = (\{u\}\cup W)\cap S$.
%%\begin{itemize}
%%\item If $W_S\neq\emptyset$ and $|W|\geq2$ then $\optdir{V_{u}}{W}=\displaystyle\bigcup_{u'\in\vartriangleleft u}{\optdir{V_{u'}}{W\cap N(u')}}$.
%%\item If $W_S\neq\emptyset$ and $W=\{w\}$ then $\optdir{V_{u}}{W}=\displaystyle\max_{\text{weight}}\left\{\bigcup_{u'\in\vartriangleleft u}{\optdir{V_{u'}}{\{w\}\cap N(u')}},\{u\}\cup\bigcup_{u'\in\vartriangleleft uw}{\optdir{V_{u'}}{\{u,w\}\cap N(u')}}\right\}$.\\
%%\item If $W_S=\emptyset$ then
%%$\optdir{V_{u}}{W}=\displaystyle\max_{\text{weight}}\left\{\bigcup_{u'\in\vartriangleleft u}{\optdir{V_{u'}}{W\cap N(u')}},\quad\{u\}\cup\bigcup_{u'\in\vartriangleleft u}{\optdir{V_{u'}}{F_{\leq2}((\{u\}\cup W)\cap N(u'))}}\right\}$.
%%\end{itemize}
\end{lemma}
%\end{restatable}
\begin{proof}
%If $u\notin\optdir{V_{u}}{W}$, then $\displaystyle\optdir{V_{u}}{W}=\bigcup_{u'\in\vartriangleleft u}{\optdir{V_{u'}}{W\cap N(u')}}$ by Lemma~\ref{lem:together}~(i). Assume that $u\in\optdir{V_{u}}{W}$. We distinguish the following three cases.

Assume that some vertex of $\{u\}\cup W$ is in $S$ and $|W|\geq2$. Then there are $w_{1},w_{2}\in W$ such that $\{u,w_{1},w_{2}\}\cap S\neq\emptyset$. Since $\{u\}\cup W$ induces a clique, we have that $\{u,w_{1},w_{2}\}$ induces an $S$-triangle, contradicting that $u$ belongs to $\optdir{V_{u}}{W}$.
Thus $W=\{w\}$ because $W \neq \emptyset$.
%Assume that some vertex of $\{u\}\cup W$ is in $S$ and $W=\{w\}$.
Then $\optdir{V_{u}}{\{w\}}=\{u\}\cup\optdir{V_{u}\setminus\{u\}}{\{u,w\}}$ by definition.
%Since $\vartriangleleft uw = \max_{G}(V_{u}\setminus(N[u]\cap N(w)))$, any vertex $u' \in \vartriangleleft uw$ is adjacent to at most one of $u$ and $w$.
Observe that for any $u'\in V_{u}\cap N(u)\cap N(w)$, the vertex set $\{u',u,w\}$ induces an $S$-triangle, since $u$ and $w$ are adjacent.
Thus, no vertex of $V_{u}\cap N(u)\cap N(w)$ is in $\optdir{V_{u}}{\{w\}}$.
%% Proposition~
%
%According to Lemma~\ref{lem:niceofundirected}, the collection $\mathcal{V}_{\vartriangleleft uw}$ is a nice partition of $V_{u}\setminus(N[u]\cap N(w))$.
%Thus, by Observation~\ref{obs:AXY}~(\ref{obs:XY:nice}) we get the desired formula.
%%:
%%$\optdir{V_{u}}{\{w\}}=\{u\}\cup\bigcup_{u'\in<\vartriangleleft uw}{\optdir{V_{u'}}{\{u,w\}\cap N(u')}}$.

Assume that no vertex of $\{u\}\cup W$ is in $S$.
Consider a vertex $u'\in\vartriangleleft u$. Observe that for any two vertices $a\in V_{u'}$ and $b\in V_{G}\setminus V_{u'}$ to be adjacent, since $r(a)\leq r(u')<r(b)$ already holds, $l<r(a)$ must hold for some $l\in L(T_{b})$.
Let $W'=(\{u\}\cup W)\cap N(u')$ and $F=F_{\leq2}(W')$. We will show that $F$ is a representation of $W'$ on $V_{u'}$.
\begin{itemize}
\item Assume there are two vertices $u_{1}'',u_{2}''\in V_{u'}$ and a vertex $w'\in W'$ such that $\{u_{1}'',u_{2}'',w'\}$ induces an $S$-triangle. Then $u_{1}'',u_{2}''$ are adjacent, so without loss of generality we may assume that $r(u_{1}'')<r(u_{2}'')$. Let $l$ be a node of $L(T_{w'})$ such that $l<r(u_{1}'')$. There is a vertex $w''\in F$ such that $l'\leq l$ for some $l'\in L(T_{w''})$. This implies that the set $\{u_{1}'',u_{2}'',w''\}$ also induces an $S$-triangle.
\item Assume there is a vertex $u''\in V_{u'}$ and two vertices $w_{1}',w_{2}'\in W'$ such that $\{u'',w_{1}',w_{2}'\}$ induces an $S$-triangle. Let $l_{1}$ and $l_{2}$ be nodes of $L(T_{w_{1}'})$ and $L(T_{w_{2}'})$ respectively such that $l_{1},l_{2}<r(u'')$.
Then, there are two distinct vertices $w_{1}'',w_{2}''\in F$ such that $l_{1}'\leq l_{1}$ and $l_{2}'\leq l_{2}$ for some $l_{1}'\in L(T_{w_{1}''})$ and some $l_{2}'\in L(T_{w_{2}''})$. This implies that the set $\{u'',w_{1}'',w_{2}''\}$ also induces an $S$-triangle.
%\todo[color=blue!40,inline]{Showing the not difficult claim involves a number of cases. Should we need to show it, it should precede this proof.}
\end{itemize}
We conclude that $\optdir{V_{u'}}{W'}=\optdir{V_{u'}}{F}$.
%Now observe that the collection $\mathcal{V}_{\vartriangleleft u}$ is a nice partition of $V_{u}\setminus\{u\}$ by Lemma~\ref{lem:niceofundirected}.
%Thus, Observation~\ref{obs:AXY}~(\ref{obs:XY:nice}) and Lemma~\ref{prop:neighb} imply the corresponding formula.
\end{proof}

We next show that Lemma~\ref{lem:together}~(ii) and Lemma~\ref{lem:Aleaf} are enough to develop a dynamic programming scheme.
As the size of the representatives is bounded with respect to $\ell$ by Lemma~\ref{lem:LL-1}, we are able to store a bounded number of partial subsolutions.
In particular we show that we only need to compute $\optdir{X}{Y}$ such that $|X| = O(n)$ and $|Y| \leq 2\ell+1$.

%\begin{restatable}[*]{theorem}{theoleafage}\label{theo:leafage}
\begin{theorem}\label{theo:leafage}
There is an algorithm that, given a connected chordal graph $G$ with leafage $\ell \geq 2$ and an expanded tree model of $G$, solves \textsc{Subset Feedback Vertex Set} in $O(n^{2\ell+1})$ time.
%\textsc{Subset Feedback Vertex Set} can be solved on chordal graphs with leafage $\ell\geq2$ in $O(n^{2\ell+1})$ time.
\end{theorem}
%\end{theorem}
%\end{restatable}
\begin{proof}
Let  $\{T, \{T_v\}_{v\in V_G}\}$ be an expanded tree model of $G$ and let $r$ be the root of $T$.
Our task is to solve SFVS by computing $\optdir{V_{r}}{\emptyset}$.
To do so, we construct a dynamic programming algorithm that visits the nodes on $T$ in a bottom-up fashion, starting from the leaves and moving towards the root $r$.
At each node $u$ of $T$, we store the values corresponding to
$\optdir{V_{u}}{\emptyset}$ and $\optdir{V_{u}}{W}$
for every $W\subseteq V_{G}\setminus V_{u}$ such that
$W\neq\emptyset$, $G[W]\in\mathcal{F}_{S}$, and $G[\{u\}\cup W]$ is a clique.
In order to compute $\optdir{V_{u}}{\emptyset}$, we apply Lemma~\ref{lem:together}~(ii) by collecting all corresponding values on the necessary descendants of $u$.
For computing $\optdir{V_{u}}{W}$, we apply Lemma~\ref{lem:Aleaf} by looking at the values stored on the necessary descendants of $u$.
%Let $u\in V_{G}$ and $W\subseteq V_{G}\setminus V_{u}$ such that $W\neq\emptyset$, $G[\{u\}\cup W]$ is a clique, and $G[W]\in\mathcal{F}_{S}$, and let $W_S = (\{u\}\cup W)\cap S$.
In particular, we deduce the following formulas, where $W_S = (\{u\}\cup W)\cap S$:
%%Then, $\optdir{V_{u}}{W}=$
\begin{itemize}
\item If $W_S\neq\emptyset$ and $|W|\geq2$ then $\optdir{V_{u}}{W}=\displaystyle\bigcup_{u'\in\vartriangleleft u}{\optdir{V_{u'}}{W\cap N(u')}}$. \\
Lemma~\ref{lem:Aleaf} implies that $u\notin\optdir{V_{u}}{W}$. Thus by Lemma~\ref{lem:together}~(i) we get the claimed formula.
\item If $W_S\neq\emptyset$ and $W=\{w\}$ then $\optdir{V_{u}}{W}=\displaystyle\max_{\text{weight}}\left\{\bigcup_{u'\in\vartriangleleft u}{\optdir{V_{u'}}{\{w\}\cap N(u')}},\{u\}\cup\bigcup_{u'\in\vartriangleleft uw}{\optdir{V_{u'}}{\{u,w\}\cap N(u')}}\right\}$.
    According to Lemma~\ref{lem:niceofundirected}, the collection $\mathcal{V}_{\vartriangleleft uw}$ is a nice partition of $V_{u}\setminus(N[u]\cap N(w))$.
Thus, by Observation~\ref{obs:AXY} we get the desired formula.
\item If $W_S=\emptyset$ then
$\optdir{V_{u}}{W}=\displaystyle\max_{\text{weight}}\left\{\bigcup_{u'\in\vartriangleleft u}{\optdir{V_{u'}}{W\cap N(u')}},\quad\{u\}\cup\bigcup_{u'\in\vartriangleleft u}{\optdir{V_{u'}}{F_{\leq2}((\{u\}\cup W)\cap N(u'))}}\right\}$. \\
According to Lemma~\ref{lem:niceofundirected}, the collection $\mathcal{V}_{\vartriangleleft u}$ is a nice partition of $V_{u}\setminus\{u\}$. Thus, Observation~\ref{obs:AXY} and Lemma~\ref{prop:neighb} imply the corresponding formula.
\end{itemize}
Notice that we compute $\optdir{V_{u}}{W}$ for $n^{O(\ell)}$ sets $W_1, \ldots, W_t \subseteq V_{G}\setminus V_{u}$
such that $W$ is represented by a set $W_i$ (i.e., there exists $W_i$ such that $\optdir{V_{u}}{W} = \optdir{V_{u}}{W_i}$).
At the root $r$ of $T$, we only compute $\optdir{V_{r}}{\emptyset}$ by applying Lemma~\ref{lem:together}~(ii).

Regarding the correctness of the algorithm, we show that the recursive formulas given in Lemma~\ref{lem:together}~(ii) and Lemma~\ref{lem:Aleaf} require only sets that are also computed via these formulas.
The formula given in Lemma~\ref{lem:together}~(ii) requires sets $\optdir{V_{u'}}{W'}$ where $u'<u$ and either $W'=\emptyset$ or $W'=\{w'\}$ such that $u'<w'$ and $u'w'\in E_{G}$. In the second case, it is not difficult to see that $u'$ and $W'$ satisfy the hypothesis of Lemma~\ref{lem:Aleaf} as they induce a graph in $\mathcal{F}_{S}$.
Notice that an induced subgraph of a clique is also a clique and an induced subgraph of a graph in $\mathcal{F}_{S}$ is also a graph in $\mathcal{F}_{S}$.
The formulas given in Lemma~\ref{lem:Aleaf} require sets $\optdir{V_{u'}}{W'}$ of the following three cases:
\begin{itemize}
\item Sets such that $u'\in\vartriangleleft u$ and $W'=W\cap N(u')\subseteq W$. These sets are clearly computed via the formulas of Lemma~\ref{lem:together}~(ii) or \ref{lem:Aleaf} according to whether $W$ is empty or not.
\item Sets such that $u'\in\vartriangleleft uw$ and $W'=\{u,w\}\cap N(u')$. Since $u'$ is only adjacent to at most one of $u$ and $w$, we have either $W'=\emptyset$ or $W'=\{w'\}$ such that $u'<w'$ and $u'w'\in E_{G}$.
\item Sets such that $u'\in\vartriangleleft u$ and $W'=F_{\leq2}((\{u\}\cup W)\cap N(u'))\subseteq(\{u\}\cup W)\cap N(u')\subseteq\{u\}\cup W$ where $(\{u\}\cup W)\cap S=\emptyset$. Since $G[\{u\}\cup W]\in\mathcal{F}_{S}$ and is a clique, we obtain that $G[W']\in\mathcal{F}_{S}$ and $G[\{u'\}\cup W']$ is a clique.
\end{itemize}
We conclude that in all cases the sets required by a formula of Lemma~\ref{lem:Aleaf} are computed via a formula given in Lemma~\ref{lem:together}~(ii) or Lemma~\ref{lem:Aleaf}.

We now analyze the running time of our algorithm.
We begin by determining for every pair of nodes $x,y$ of $T$ whether $x<y$, $y<x$ or they are incomparable.
Since for any one pair this can be done in $O(n)$ time, we complete this task in $O(n^{3})$ time.
Notice that, since the input graph has leafage $\ell$, any subset of $V_{T}$ of pairwise incomparable nodes is of size at most $\ell$.
This fact implies that $|F_{\leq2}(U)|\leq 2\ell$ for any $U\subseteq V_{G}$.
Due to the recursion, we only need to compute $F_{\leq2}(U)$ for sets $U$ such that $|U|\leq 2\ell+1$.
Computing any such set requires at most $(2\ell+1)^{2}$ comparisons and consequently constant time, so the total preprocessing time is $O(n^{2\ell+1})$.
Now consider a set $\optdir{X}{Y}$.
The parts of any partition of $X$ that we use in our formulas are rooted in pairwise incomparable nodes.
This means that any set $\optdir{X}{Y}$ is computed in $O(\ell)$ time.
Thus we conclude that total running time of our algorithm is $O(n^{2\ell+1})$.
%We conclude that the total running time of our algorithm is $O(n^{2\ell+1})$.
\end{proof}

%\medskip
%% perhaps explain the disconnected case.
If we let the leafage of a chordal graph to be the maximum over all of its connected components then we reach to the following result.

%\begin{restatable}[*]{corollary}{corleafage}\label{cor:leafage}
\begin{corollary}\label{cor:leafage}
\textsc{Subset Feedback Vertex Set} can be solved in time $n^{O(\ell)}$ for chordal graphs with leafage at most $\ell$.
\end{corollary}
%\end{restatable}
\begin{proof}
For every connected component $C$ of a chordal graph $G$, we compute its leafage and the corresponding tree model $\mathcal{T}(C)$ by using the $O(n^3)$-time algorithm of Habib and Stacho~\cite{HabibS09}.
If the leafage of $C$ is less than two, then $C$ is an interval graph and we can compute $\optdir{\emptyset}{V(C)}$ in $n^{O(1)}$ time by running the algorithm for SFVS on interval graphs given in \cite{PapT19}.
Otherwise, we compute the expanded tree model $\mathcal{T'}(C)$ from $\mathcal{T}(C)$ by Lemma~\ref{lem:expanded} in $O(n^2)$ time.
Applying Theorem~\ref{theo:leafage} on $\mathcal{T'}(C)$ shows that $\optdir{V(C)}{\emptyset}$ can be computed in $n^{O(\ell)}$.
Since the connected components of $G$ form a nice partition of $V(G)$, Observation~\ref{obs:AXY} implies that $\optdir{V(G)}{\emptyset}$ is the union of all $\optdir{V(C)}{\emptyset}$ for every connected component $C$ of $G$.
Therefore all steps can carried out in $n^{O(\ell)}$ time.
\end{proof}

We next prove that we can hardly avoid the dependence of the exponent in the stated running time, since we show that \textsc{Subset Feedback Vertex Set} is W[1]-hard parameterized by the leafage of a chordal graph.
Our reduction is inspired by the W[1]-hardness of \textsc{Feedback Vertex Set} parameterized by the mim-width given by Jaffke et al.~in \cite{JaffkeKT20}.

\begin{theorem}\label{theo:whard}
\textsc{Subset Feedback Vertex Set} on chordal graphs is W[1]-hard when parameterized by its leafage.
\end{theorem}
\begin{proof}
We provide a reduction from the \textsc{Multicolored Clique} problem. Given a graph $G=(V,E)$ and a partition $\{V_{i}\}_{i\in[k]}$ of $V$ into $k$ parts, the \textsc{Multicolored Clique} (MCC) problem asks whether $G$ has a clique that contains exactly one vertex of $V_{i}$ for every $i\in[k]$. It is known that MCC is W[1]-hard when parameterized by $k$ \cite{FellowsHRV09,Pietrzak03}.

\begin{figure}\label{fig:subtree}
\begin{center}
\begin{tikzpicture}
\tikzstyle{b}=[circle,inner sep=2pt,draw,fill]
\tikzstyle{p}=[decorate,decoration={snake,amplitude=.4mm,segment length=2mm,post length=0mm}]
%nodes
\draw(0,0)node[b](r){}node[anchor=south]{$r$}++(0,-1)node[b](y){}node[anchor=north]{$y_{ij}$};
\draw(-1,0)node[b](xi1){}node[anchor=south]{$x_{i}^{1}$}++(-1,0)node[b](xi2){}node[anchor=south]{$x_{i}^{2}$}++(-1,0)node[b](xi3){}node[anchor=south]{$x_{i}^{p-1}$}++(-1,0)node[b](xi4){}node[anchor=south]{$x_{i}^{p}$};
\draw(-1,-1)node[b](xim1){}node[anchor=north]{$x_{i}^{-1}$}++(-1,0)node[b](xim2){}node[anchor=north]{$x_{i}^{-2}$}++(-1,0)node[b](xim3){}node[anchor=north]{$x_{i}^{1-p}$}++(-1,0)node[b](xim4){}node[anchor=north]{$x_{i}^{-p}$};
\draw(1,0)node[b](xj1){}node[anchor=south]{$x_{j}^{1}$}++(1,0)node[b](xj2){}node[anchor=south]{$x_{j}^{2}$}++(1,0)node[b](xj3){}node[anchor=south]{$x_{j}^{p-1}$}++(1,0)node[b](xj4){}node[anchor=south]{$x_{j}^{p}$};
\draw(1,-1)node[b](xjm1){}node[anchor=north]{$x_{j}^{-1}$}++(1,0)node[b](xjm2){}node[anchor=north]{$x_{j}^{-2}$}++(1,0)node[b](xjm3){}node[anchor=north]{$x_{j}^{1-p}$}++(1,0)node[b](xjm4){}node[anchor=north]{$x_{j}^{-p}$};
%edges
\draw[-](r)--(y);
\draw[-](r)--(xi1)--(xi2)
             (xi3)--(xi4);
\draw[-,p](xi2)--(xi3);
\draw[-](r)--(xim1)--(xim2)
             (xim3)--(xim4);
\draw[-,p](xim2)--(xim3);
\draw[-](r)--(xj1)--(xj2)
             (xj3)--(xj4);
\draw[-,p](xj2)--(xj3);
\draw[-](r)--(xjm1)--(xjm2)
             (xjm3)--(xjm4);
\draw[-,p](xjm2)--(xjm3);
\end{tikzpicture}
\end{center}
\caption{The subtree $T({\{x_{i}^{+},x_{i}^{-},y_{ij},x_{j}^{+},x_{j}^{-}\}})$ of $T$ for some $i,j\in[k]$ such that $i<j$.}
\end{figure}
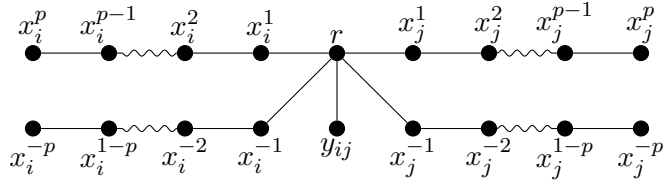

Let $(G=(V,E),\{V_{i}\}_{i\in[k]})$ be an instance of MCC. We assume that $k\geq10$ and without loss of generality that there exists $p\in\mathbb{N}$ such that $V_{i}=\{v_{i}^{j}\}_{j\in[p]}$ for every $i\in[k]$. We consider the $\frac{k}{2}(k+3)$-star $T$ with root $r$ and leaves $x_{i}^{+},x_{i}^{-}$ for every $i\in[k]$ and $y_{ij}$ for every $i,j\in[k]$ such that $i<j$. We modify the star $T$ as follows: for every $i\in[k]$, we transform the edge $\mypath{r,x_{i}^{+}}$ through edge subdivisions into $\mypath{r=x_{i}^{0},x_{i}^{1},\ldots,x_{i}^{p}=x_{i}^{+}}$ and in a similar way we replace the edge $\mypath{r,x_{i}^{-}}$ by $\mypath{r=x_{i}^{0},x_{i}^{-1},\ldots,x_{i}^{-p}=x_{i}^{-}}$.
Given a set $X$ of vertices of $T$, we write $T(X)$ to denote the minimal subtree of $T$ containing all vertices of $X$. A certain subtree of $T$ is depicted in Figure \ref{fig:subtree}. We define the following subtrees of $T$:
\begin{itemize}
\item For every $i,j\in[k]$ such that $i<j$ and for every $a,b\in[p]$ such that $v_{i}^{a}v_{j}^{b}\in E$, we define $e_{ij}^{ab}=T(\{x_{i}^{a},x_{i}^{a-p},y_{ij},x_{j}^{b},x_{j}^{b-p}\})$. We denote by $R$ the set of all these subtrees. For all $i,j\in[k]$ such that $i<j$, we denote by $R_{ij}$ the set of all subtrees in $R$ with subscript $ij$. For all $i\in[k]$:
\begin{itemize}
\item We denote the set $\{e_{ij}^{ab}\in R\}\cup\{e_{ji}^{ba}\in R\}$ by $R_{i}$.
\item For all $a_{i}\in[p]$, we denote the set $\{e_{ij}^{a_{i}b}\in R\}\cup\{e_{ji}^{ba_{i}}\in R\}$ by $R_{i}^{a_{i}}$.
\end{itemize}
\item For every $i\in[k]$ and $a\in[p]$, we define $s_{i}^{a,1}=s_{i}^{a,2}=T[\{x_{i}^{a}\}]$ and $s_{i}^{-a,1}=s_{i}^{-a,2}=T[\{x_{i}^{-a}\}]$. We denote by $S_{V}$ the set of all these subtrees. For all $i\in[k]$:
\begin{itemize}
\item We denote by $S_{i}$ the set of all subtrees in $S_{V}$ with subscript $i$.
\item For all $a_{i}\in[p]$, we denote the set $\{s_{i}^{a,c}\in S_{i}:a_{i}-p\leq a\leq a_{i}\}$ by $S_{i}^{a_{i}}$.
\end{itemize}
\item For every $i,j\in[k]$ such that $i<j$, we define $s_{ij}=T[\{y_{ij}\}]$. We denote by $S_{E}$ the set of all these subtrees.
\end{itemize}
We further denote by $S$ the set $S_{V}\cup S_{E}$ and by $\mathcal{T}$ the collection $R\cup S$. We construct a graph $G'$ that is the intersection graph of $(T,\mathcal{T})$. Notice that $G'$ is a chordal graph of leafage at most $\frac{k}{2}(k+3)$. We identify the vertices of $G'$ with their corresponding subtrees in $\mathcal{T}$.
By the construction of $(T,\mathcal{T})$, reganding adjacencies between vertices of $G'$ we observe the following:
\begin{itemize}
\item $R$ induces a clique, because all its elements contain the node $r$.
\item For every $i\in[k]$ and $a\in[p]$, the vertices $s_{i}^{a,1}$ and $s_{i}^{-a,1}$ are adjacent to $s_{i}^{a,2}$ and $s_{i}^{-a,2}$ respectively.
\item For every $i\in[k]$ and $a_{i}\in[p]$, we have $N(e)\cap S_{i}=S_{i}^{a_{i}}$ for all $e\in R_{i}^{a_{i}}$.
\item For every $i,j\in[k]$ such that $i<j$, we have $N(s_{ij})=R_{ij}$.
\end{itemize}
We set the weight of all vertices of $R$, $S_{V}$ and $S_{E}$ to be $\frac{p}{2}$, $1$ and $\frac{p}{2}m$ respectively. We will show that $(G,\{V_{i}\}_{i\in[k]})$ is a \textsc{Yes}-instance of MCC if and only if there exists a solution to SFVS on $(G',S)$ of weight $\frac{p}{2}(m-\frac{k}{2}(k-9))$.

For the forward direction, let $\{v_{1}^{a_{1}},\ldots,v_{k}^{a_{k}}\}$ be a solution to MCC on $(G,\{V_{i}\}_{i\in[k]})$. We define $R_{C}=\{e_{ij}^{a_{i}a_{j}}:i,j\in[k],i<j\}$. Observe that $R_{C}$ contains exactly one element of $R_{ij}$ for each $i,j\in[k]$ such that $i<j$. We further define the set $U=(R\setminus R_{C})\cup(\cup\{S_{i}^{a_{i}}\}_{i\in[k]})$. Now observe that in $G-U$ each of the remaining vertices of $S$ has exactly one neighbour. Thus $U$ a solution to SFVS on $(G',S)$ of weight $\frac{p}{2}(m-\frac{k}{2}(k-1))+2pk=\frac{p}{2}(m-\frac{k}{2}(k-9))$.

For the reverse direction, let $U$ be a solution to SFVS on $(G',S)$ of weight $\frac{p}{2}(m-\frac{k}{2}(k-9))$. Notice that no element of $S_{E}$ can be in $U$. Consequently, for every $i,j\in[k]$ such that $i<j$, $|R_{ij}\setminus U|\leq1$, since any two elements of $R_{ij}$ along with $s_{ij}$ form an $S$-triangle of $G'$. Any one of the remaining $S$-triangles of $G'$ is formed by
\begin{itemize}
\item either an element of $R_{i}^{a_{i}}$ and both $s_{i}^{a,1}$ and $s_{i}^{a,2}$
\item or an element of $R_{i}^{a_{i}}$, an element of $R_{i}^{a_{i}'}$ and either $s_{i}^{a,1}$ or $s_{i}^{a,2}$
\end{itemize}
for some $i\in[k]$ and for some $a,a_{i},a_{i}'\in[p]$ such that $\max\{a_{i},a_{i}'\}-p\leq a\leq\min\{a_{i},a_{i}'\}$. Let $i\in[k]$.
\begin{claim}
If $|R_{i}\setminus U|\geq 1$, then $|S_{i}\cap U|\geq p$.
\end{claim}
\begin{claimproof}
Assume that $R_{i}\setminus U\supseteq\{e\}$ and $e\in R_{i}^{a_{i}}$ for some $a_{i}\in[p]$. Then for every integer $a\neq0$ such that $a_{i}-p\leq a\leq a_{i}$ at least one of $s_{i}^{a,1},s_{i}^{a,2}$ must be in $U$, so $|S_{i}\cap U|\geq p$.
\end{claimproof}
\begin{claim}
If $|R_{i}\setminus U|\geq2$, then $|S_{i}\cap U|\geq2p$.
\end{claim}
\begin{claimproof}
Assume that $R_{i}\setminus U\supseteq\{e,e'\}$ and $e\in R_{i}^{a_{i}}$ and $e'\in R_{i}^{a_{i}'}$ for some $a_{i},a_{i}'\in[p]$ such that $a_{i}\leq a_{i}'$. Then
\begin{itemize}
\item for every integer $a\neq0$ such that $a_{i}'-p\leq a\leq a_{i}$ both $s_{i}^{a,1}$ and $s_{i}^{a,2}$ must be in $U$ and
\item for every integer $a\neq0$ such that $a_{i}-p\leq a<a_{i}'-p$ or $a_{i}<a\leq a_{i}'$ at least one of $s_{i}^{a,1},s_{i}^{a,2}$ must be in $U$,
\end{itemize}
so $|S_{i}\cap U|\geq 2(a_{i}+(p-a_{i}'))+1((a_{i}'-a_{i})+(a_{i}'-a_{i}))=2p$.
\end{claimproof}
\begin{claim}
If $|R_{i}\setminus U|\geq3$, then $|S_{i}\cap U|=2p$ only if $R_{i}\setminus U\subseteq R_{i}^{a_{i}}$ for some $a_{i}\in[p]$.
\end{claim}
\begin{claimproof}
Assume that $R_{i}\setminus U\supseteq\{e,e',e''\}$ and $e\in R_{i}^{a_{i}}$, $e'\in R_{i}^{a_{i}'}$ and $e''\in R_{i}^{a_{i}''}$ for some $a_{i},a_{i}',a_{i}''\in[p]$ such that $a_{i}\leq a_{i}'\leq a_{i}''$. Then
\begin{itemize}
\item for every integer $a\neq0$ such that $a_{i}'-p\leq a\leq a_{i}'$ both $s_{i}^{a,1}$ and $s_{i}^{a,2}$ must be in $U$ and
\item for every integer $a\neq0$ such that $a_{i}-p\leq a<a_{i}'-p$ or $a_{i}'<a\leq a_{i}''$ at least one of $s_{i}^{a,1},s_{i}^{a,2}$ must be in $U$,
\end{itemize}
so $|S_{i}\cap U|\geq 2p+1((a_{i}'-a_{i})+(a_{i}''-a_{i}'))=2p+(a_{i}''-a_{i})$. We conclude that for $|R_{i}\setminus U|$ to be $2p$, all elements of $R_{i}\setminus U$ must be elements of the same $R_{i}^{a_{i}}$ for some $a_{i}\in[p]$.
\end{claimproof}

Assume that $|\{i\in[k]:|R_{i}\setminus U|=1\}|=k'$ and $|\{i\in[k]:|R_{i}\setminus U|\geq2\}|=k''$. Then $|R\setminus U|\leq k'+\frac{k''}{2}(k''-1)$ and $|S_{V}\cap U|\geq p(k'+2k'')$, so the weight of $U$ must be at least
$$\frac{p}{2}\left(m-k'-\frac{k''}{2}(k''-1)\right)+p(k'+2k'')=\frac{p}{2}\left(m+k'-\frac{k''}{2}(k''-9)\right)=B(k',k'').$$
Clearly, $k',k''\in\{0,\ldots,k\}$. Regarding the values of $B$, we obsreve the following:
\begin{itemize}
\item $B(0,k'')\leq B(k',k'')$ for all $k',k''\in\{0,\ldots,k\}$,
\item $B(0,9)\leq B(0,k'')$ for all $k''\in\{0,\ldots,9\}$ and
\item $B(0,k'')<B(0,k''-1)$ for all $k''\in\{10,\ldots,k\}$.
\end{itemize}
We conclude that a weight of $\frac{p}{2}(m-\frac{k}{2}(k-9))=B(0,k)$ is within bounds only for $k'=0$ and $k''=k$ and can be achieved only when $|R\setminus U|=\frac{k}{2}(k-1)$ and $|S_{i}\cap U|=2p$ for all $i\in[k]$. For every $i\in[k]$, let $a_{i}\in[p]$ be such that $R_{i}\setminus U\subseteq R_{i}^{a_{i}}$. Then the set $\{v_{1}^{a_{1}},\ldots,v_{k}^{a_{k}}\}$ is a solution to MCC on $(G,\{V_{i}\}_{i\in[k]})$.
\end{proof}

\section{SFVS on rooted path graphs}\label{sec:dipath}
%\todo[inline]{try to add the old formulas within the proofs}
Here we show how to extend our previous approach for SFVS on rooted path graphs.
Rooted path graphs are exactly the intersection graphs of rooted paths on a rooted tree.
Notice that rooted path graphs have unbounded leafage.
Our main goal is to derive a recursive formulation for $\optdir{X}{Y}$, similar to Lemma~\ref{lem:Aleaf}.
%In particular, we show that it is sufficient to consider sets $X$ and $Y$ such that $|X|=O(n+m)$ and $|Y|\leq 1$.
In particular, we show that it is sufficient to consider sets $Y$ containing at most one vertex.

For any vertex $u$ of $G$, we denote the leaf of its corresponding rooted path in $T$ by $l(u)$.
We need to define further special vertices and subsets.
Let $u,v\in V_{G}$ such that $u<v$. The (unique) maximal predecessor $u'$ of $v$ such that $l(u')<r(u)\leq r(u')$ is denoted by $u\vartriangleleft v$.
Moreover, for every $V_{1},V_{2},V_{3}\subseteq V_{G}$, we define the following sets:
\begin{itemize}
\item[$-$] \hspace*{-0.2cm} $V_{\langle V_{1}|V_{2}|V_{3}\rangle}=\{u\in V_{G}:r(v_{1})<l(u)<r(v_{2})<r(u)\leq r(v_{3})\text{ for some }v_{i}\in V_{i},i\in\{1,2,3\}\}$
\item[$-$] \hspace*{-0.2cm} $V_{\langle|V_{2}|V_{3}\rangle}=\{u\in V_{G}:l(u)<r(v_{2})<r(u)\leq r(v_{3})\text{ for some }v_{i}\in V_{i},i\in\{2,3\}\}$
\item[$-$] \hspace*{-0.2cm} $V_{\langle V_{1}||V_{3}\rangle}=\{u\in V_{G}:r(v_{1})<l(u)<r(u)\leq r(v_{3})\text{ for some }v_{i}\in V_{i},i\in\{1,3\}\}$
\end{itemize}
The vertical bars indicate the placements of $l(u)$ and $r(u)$ with respect to $V_{1},V_{2},V_{3}$.

%\begin{restatable}[*]{lemma}{lemniceofdirected}\label{lem:niceofdirected}
\begin{lemma}\label{lem:niceofdirected}
Let $u,w\in V_{G}\setminus S$ such that $u<w$ and $uw\in E_{G}$. Then, the collection
\begin{displaymath}
\displaystyle\{V_{\langle\vartriangleleft uw || \vartriangleleft u\rangle}\setminus S\}\cup\{V_{u'}\cup(V_{\langle|\{u'\}|\{u'\vartriangleleft u\}\rangle}\setminus S)\}_{u'\in\vartriangleleft uw}
%\ $\displaystyle\{V_{\langle\vartriangleleft uw || \vartriangleleft u\rangle}\setminus S\}\cup\{V_{u'}\cup(V_{\langle|\{u'\}|\{u'\vartriangleleft u\}\rangle}\setminus S)\}_{u'\in\vartriangleleft uw}$ \
\end{displaymath}
is a nice partition of $X=V_{u}\setminus(\{u\}\cup(N(u)\cap N(w)\cap S))$ with respect to any $Y\subseteq V_{G}\setminus X$ such that $Y\cap S=\emptyset$.
\end{lemma}
%\end{restatable}
\begin{proof}
We first show that this collection is indeed a partition of $X$. Recall that the vertices of $\vartriangleleft uw$ induce an independent set by Lemma~\ref{prop:disc}.
Consider a vertex $u''\in X$. Then exactly one of the following statements holds, implying the claimed partition:
\begin{displaymath}
\left\{\begin{array}{l}
\exists u'\in\vartriangleleft uw:r(u'')\leq r(u')\\
\exists u'\in\vartriangleleft uw:l(u'')<r(u')<r(u'')\\
\exists u'\in\vartriangleleft uw:r(u')<l(u'')
\end{array}\right.\Rightarrow\left\{\begin{array}{l}
\exists u'\in\vartriangleleft uw:u''\in V_{u'}\\
\exists u'\in\vartriangleleft uw:u''\in V_{(\langle|\{u'\}|\{u'\vartriangleleft u\}\rangle}\setminus S\\
u''\in V_{\langle\vartriangleleft uw||\vartriangleleft u\rangle}\setminus S
\end{array}\right.
\end{displaymath}

Now let $Y\subseteq V_{G}\setminus X$ such that $Y\cap S=\emptyset$ and consider an $S$-triangle $S_{t}$ of $G[X\cup Y]$. Then there is some vertex $u'\in\vartriangleleft uw$ such that $V(S_{t})\cap V_{u'}\cap S\neq\emptyset$. Since $S_{t}$ is a triangle, every vertex in $V(S_{t})\setminus V_{u'}$ must be adjacent to every vertex in $V(S_{t})\cap V_{u'}$. By Lemma~\ref{prop:neighb}, a vertex $u''\in V_{G}\setminus V_{u'}$ that is adjacent to a vertex of $V_{u'}$ must be adjacent to $u'$, so $l(u'')<r(u')<r(u'')$ must hold.
We conclude that $V(S_{t})\subseteq V_{u'}\cup(V_{\langle|\{u'\}|\{u'\vartriangleleft u\}\rangle}\setminus S)$.
\end{proof}

For every appropriate $u,v$, we will denote the set $V_{u}\cup(V_{\langle|\{u\}|\{v\}\rangle}\setminus S)$ by $V_{u,v}$.
Observe that the set $V_{u,u}$ is simply $V_{u}$.
First we consider the set $\optdir{V_{u}}{\{w\}}$ for which $u<w$ and $uw\in E_{G}$.
%We distinguish two cases depending on whether $u$ belongs to $\optdir{V_{u}}{\{w\}}$ or not.
\begin{itemize}
\item If $u\notin\optdir{V_{u}}{\{w\}}$ then $\optdir{V_{u}}{\{w\}}=\bigcup_{u'\in\vartriangleleft u}{\optdir{V_{u'}}{\{w\}\cap N(u')}}$ by Lemma~\ref{lem:together}~(i).
\end{itemize}
Also, recall that $\optdir{V_{u}}{\emptyset}$ is described by the formula given in Lemma~\ref{lem:together}~(ii). 
We derive the following result that handles the cases in which $u\in\optdir{V_{u}}{\{w\}}$.
%\todo[inline,color=blue!40]{Here (or, more appropriately, at the beginning of this subsection) I would like something in the terms of `We keep the lemma regarding $\optdir{V_{u}}{\emptyset}$ as is. However, in directed path graphs, for the remaining cases, we can do better'.}

%\begin{restatable}[*]{lemma}{lemAdiruw}\label{lem:Adiruw}
\begin{lemma}\label{lem:Adiruw}
Let $u,w\in V_{G}$ such that $u<w$ and $uw\in E_{G}$, and let $u\in\optdir{V_{u}}{\{w\}}$. %Then,
%$\optdir{V_{u}}{\{w\}}=$
%Assume that $u\in\optdir{V_{u}}{\{w\}}$. %Then, the following hold
\begin{itemize}
\item If $u\in S$ or $w\in S$ then $\displaystyle\optdir{V_{u}}{\{w\}}=\{u\}\cup\bigcup_{u'\in\vartriangleleft uw}{\optdir{V_{u'}}{\{u,w\}\cap N(u')}}$.
\item If $u,w\notin S$ then $\displaystyle\optdir{V_{u}}{\{w\}}=\{u\}\cup(V_{\langle\vartriangleleft uw||\vartriangleleft u\rangle}\setminus S)\cup\bigcup_{u'\in\vartriangleleft uw}{\optdir{V_{u',u'\vartriangleleft u}}{\{u,w\}\cap N(u')}}$.
\end{itemize}
\end{lemma}
%\end{restatable}
\begin{proof}
%We distinguish two cases depending on whether $u$ belongs to $\optdir{V_{u}}{\{w\}}$ or not.
%If $u\notin\optdir{V_{u}}{\{w\}}$ then $\optdir{V_{u}}{\{w\}}=\bigcup_{u'\in\vartriangleleft u}{\optdir{V_{u'}}{\{w\}\cap N(u')}}$ by Lemma~\ref{lem:together}~(i).
%
%Assume that $u\in\optdir{V_{u}}{\{w\}}$.
Observe that $\optdir{V_{u}}{\{w\}}=\{u\}\cup\optdir{V_{u}\setminus\{u\}}{\{u,w\}}$ by definition.
Since $\vartriangleleft uw = \max_{G}(V_{u}\setminus(N[u]\cap N(w)))$, any vertex $u' \in \vartriangleleft uw$ is adjacent to at most one of $u$ and $w$.
Regarding triangles of $G[V_{u}\cup\{w\}]$, we observe the following property:
\begin{description}
\item[(P1)] By the hypothesis, the vertices $u$ and $w$ are adjacent. Thus, for any $u'\in V_{u}\cap N(u)\cap N(w)$, the vertex set $\{u',u,w\}$ induces a triangle.
\end{description}

If $u\in S$ or $w\in S$, then no vertex of $V_{u}\cap N(u)\cap N(w)$ is in $\optdir{V_{u}}{\{w\}}$ because of (P1).
%% Proposition~
According to Lemma~\ref{lem:niceofundirected}, the collection $\mathcal{V}_{\vartriangleleft uw}$ is a nice partition of $V_{u}\setminus(N[u]\cap N(w))$.
Thus, by Observation~\ref{obs:AXY} we get the desired formula:
$\optdir{V_{u}}{\{w\}}=\{u\}\cup\bigcup_{u'\in<\vartriangleleft uw}{\optdir{V_{u'}}{\{u,w\}\cap N(u')}}$.

If $u,w\notin S$, then no vertex of $V_{u}\cap N(u)\cap N(w)\cap S$ is in $\optdir{V_{u}}{\{w\}}$ because of (P1).
Thus, by definition, we have $\optdir{V_{u}\setminus\{u\}}{\{u,w\}}=\optdir{(V_{u}\setminus\{u\})\setminus(N(u)\cap N(w)\cap S)}{\{u,w\}}$.
Let $X=(V_{u}\setminus\{u\})\setminus(N(u)\cap N(w)\cap S)$.
By Lemma~\ref{lem:niceofdirected}, the collection $\{V_{\langle\vartriangleleft uw||\vartriangleleft u\rangle}\setminus S\}\cup\{V_{u',u'\vartriangleleft u}\}_{u'\in\vartriangleleft uw}$ is a nice partition of $X$ with respect to $\{u,w\}$.
Thus, Observation~\ref{obs:AXY} and Lemma~\ref{prop:neighb} imply
\begin{displaymath}
\begin{array}{lcl}
\optdir{X}{\{u,w\}}&=&\optdir{V_{\langle\vartriangleleft uw||\vartriangleleft u\rangle}\setminus S}{\{u,w\}}\cup\bigcup_{u'\in\vartriangleleft uw}{\optdir{V_{u',u'\vartriangleleft u}}{\{u,w\}}}\\[6pt]
&=& (V_{\langle\vartriangleleft uw||\vartriangleleft u\rangle}\setminus S)\cup\bigcup_{u'\in\vartriangleleft uw}{\optdir{V_{u',u'\vartriangleleft u}}{\{u,w\}\cap N(u')}}.
\end{array}
\end{displaymath}
\end{proof}
%
%\begin{corollary}\label{cor:rootone}
%Let $u,w\in V_{G}$ such that $u<w$ and $uw\in E_{G}$. 
%\begin{itemize}
%\item If $u\in S$ or $w\in S$ then \\
%\hspace*{0.01in} $\optdir{V_{u}}{\{w\}}=\displaystyle\max_{weight}\left\{\bigcup_{u'\in\vartriangleleft u}{\optdir{V_{u'}}{\{w\}\cap N(u')}},\quad\{u\}\cup\bigcup_{u'\in\vartriangleleft uw}{\optdir{V_{u'}}{\{u,w\}\cap N(u')}}\right\}$.
%\item If $u,w\notin S$ then \\
%\hspace*{0.01in} $\optdir{V_{u}}{\{w\}}=\displaystyle\max_{weight}\left\{\bigcup_{u'\in\vartriangleleft u}{\optdir{V_{u'}}{\{w\}\cap N(u')}},\quad\{u\}\cup(V_{\langle\vartriangleleft uw||\vartriangleleft u\rangle}\setminus S)\cup\bigcup_{u'\in\vartriangleleft uw}{\optdir{V_{u',u'\vartriangleleft u}}{\{u,w\}\cap N(u')}}\right\}$.
%\end{itemize}
%\end{corollary}

We next deal with the sets $\optdir{V_{u,v}}{Y}$ for which $u<v$, $|Y|\leq 1$ and no vertex of $\{v\}\cup Y$ belongs to $S$.
Observe that $V_{u,v}$ is not necessarily described by a set $V_{w}$ for some $w\in V_G$.
Thus we need appropriate formulas that handle such sets.
%\todo[inline,color=red!40]{The lemma for $A_{V_{u,v}}^{\emptyset}$ was (and should be) here.}
For doing so, notice that
\begin{itemize}
\item $V_{u,v}\setminus\{v\}=V_{u,u\vartriangleleft v}$, since $V_{u,v}=V_{u}\cup(V_{\langle|\{u\}|\{v\}\rangle}\setminus S)$ and $u\leq u{\vartriangleleft}v< v$.
\end{itemize}
This means that if $v\notin\optdir{V_{u,v}}{Y}$, we have $\optdir{V_{u,v}}{Y}=\optdir{V_{u,u\vartriangleleft v}}{Y}$.
%
%\begin{corollary}\label{cor:roottwo}
%Let $u\in V_{G}$ and $v\in V_{G}\setminus S$ such that $u<v$ and $uv\in E_{G}$. Then,
%\begin{displaymath}
%\optdir{V_{u,v}}{\emptyset}=\max_{weight}\left\{\optdir{V_{u,u\vartriangleleft v}}{\emptyset},\quad\{v\}\cup\optdir{V_{u,u\vartriangleleft v}}{\{v\}}\right\}.
%\end{displaymath}
%\end{corollary}

With the next result we consider the corresponding case in which $v\in\optdir{V_{u,v}}{\{w\}}$.
Notice that given a partition $\mathcal{P}$ of a set $X$ and a set $X'\subseteq X$, the collection $\mathcal{P}'=\{P\cap X'\}_{P\in\mathcal{P}}$ is a partition of $X'$.
Furthermore, observe that if $\mathcal{P}$ is a nice partition of $X$ with respect to a set $Y\subseteq V_{G}\setminus X$ such that $Y\cap S=\emptyset$, then $\mathcal{P}'$ is a nice partition of $X'$ with respect to $Y$.

%\begin{restatable}[*]{lemma}{lemAdirutovandw}\label{lem:Adirutovandw}
\begin{lemma}\label{lem:Adirutovandw}
Let $u\in V_{G}$ and $v,w\in V_{G}\setminus S$ such that $u<v<w$ and $\{u,v,w\}$ induce a clique and let $v\in\optdir{V_{u,v}}{\{w\}}$.
Then, $\displaystyle\optdir{V_{u,v}}{\{w\}}=\{v\}\cup(V_{\langle\vartriangleleft vw|\{u\}|\{u\vartriangleleft v\}\rangle}\setminus S)\cup\bigcup_{u'\in V_{u}\cap\vartriangleleft vw}{\optdir{V_{u',u'\vartriangleleft v}}{\{v,w\}\cap N(u')}}$.
\end{lemma}
%\end{restatable}
\begin{proof}
%We distinguish two cases depending on whether $v$ is in $\optdir{V_{u,v}}{\{w\}}$ or not. Let $v\notin\optdir{V_{u,v}}{\{w\}}$. Then $\optdir{V_{u}}{\{w\}}=\optdir{V_{u}\setminus\{v\}}{\{w\}}=\optdir{V_{u,u\vartriangleleft v}}{\{w\}}$. Now let $v\in\optdir{V_{u,v}}{\{w\}}$.
By definition, we have $\optdir{V_{u,v}}{\{w\}}=\{v\}\cup\optdir{V_{u,v}\setminus\{v\}}{\{v,w\}}=\{v\}\cup\optdir{V_{u,u\vartriangleleft v}}{\{v,w\}}$. Regarding triangles of $G[V_{u,v}\cup\{w\}]$, we observe the following property:
\begin{description}
\item[(P2)] By the hypothesis, the vertices $v$ and $w$ are adjacent. Thus, for any $u'\in V_{u,v}\cap N(v)\cap N(w)$, the vertex set $\{u',v,w\}$ induces a triangle.
\end{description}
Since $v,w\notin S$, no vertex of $V_{u,v}\cap N(v)\cap N(w)\cap S$ is in $\optdir{V_{u,v}}{\{w\}}$ because of (P2).
Thus $\optdir{V_{u,u\vartriangleleft v}}{\{v,w\}}=\optdir{V_{u,u\vartriangleleft v}\setminus(N(v)\cap N(w)\cap S)}{\{v,w\}}$.
Let $X=(V_{v}\setminus\{v\})\setminus(N(v)\cap N(w)\cap S)$ and $X'=V_{u,u\vartriangleleft v}\setminus(N(v)\cap N(w)\cap S)$.
Now, notice that $X'\subseteq X$.
Applying Lemma~\ref{lem:niceofdirected} on $X$ and $Y=\{v,w\}$ shows that the collection $\{V_{\langle\vartriangleleft vw|\{u\}|\{u\vartriangleleft v\}\rangle}\setminus S\}\cup\{V_{u',u'\vartriangleleft v}\}_{u'\in V_{u}\cap\vartriangleleft vw}$ is a nice partition of $X'$ with respect to $Y$.
Hence, Observation~\ref{obs:AXY} and Lemma~\ref{prop:disc} imply
\begin{displaymath}
\begin{array}{lcl}
\optdir{X'}{\{v,w\}}&=&\optdir{V_{\langle\vartriangleleft vw|\{u\}|\{u\vartriangleleft v\}\rangle}\setminus S}{\{v,w\}}\cup\bigcup_{u'\in V_{u}\cap\vartriangleleft vw}{\optdir{V_{u',u'\vartriangleleft v}}{\{v,w\}}}\\[6pt]
                    &=&(V_{\langle\vartriangleleft vw|\{u\}|\{u\vartriangleleft v\}\rangle}\setminus S)\cup\bigcup_{u'\in V_{u}\cap\vartriangleleft vw}{\optdir{V_{u',u'\vartriangleleft v}}{\{v,w\}\cap N(u')}}.
\end{array}
\end{displaymath}
\end{proof}
%
%\begin{corollary}\label{cor:rootthree}
%Let $u\in V_{G}$ and $v,w\in V_{G}\setminus S$ such that $u<v<w$ and $\{u,v,w\}$ induce a clique. Then, $\displaystyle\optdir{V_{u,v}}{\{w\}}=\displaystyle\max_{weight}\left\{\optdir{V_{u,u\vartriangleleft v}}{\{w\}},\quad\{v\}\cup(V_{\langle\vartriangleleft vw|\{u\}|\{u\vartriangleleft v\}\rangle}\setminus S)\cup\bigcup_{u'\in V_{u}\cap\vartriangleleft vw}{\optdir{V_{u',u'\vartriangleleft v}}{\{v,w\}\cap N(u')}}\right\}$.
%\end{corollary}

Now we are in position to state our claimed result, which is obtained in a similar fashion with the algorithm given in Theorem~\ref{theo:leafage}.

%\begin{restatable}[*]{theorem}{theodirectedP}\label{theo:directedP}
\begin{theorem}\label{theo:directedP}
\textsc{Subset Feedback Vertex Set} can be solved on rooted path graphs in $O(n^{2}m)$ time.
\end{theorem}
%\end{restatable}
\begin{proof}
We first describe the algorithm. Given a rooted path graph $G$, we construct its tree model $\{T, \{T_v\}_{v\in V_G}\}$ in $O(n+m)$ time \cite{Dietz84,Gavril75}.
If $G$ is an interval graph then SFVS can be solved by the algorithm described in \cite{PapT19} that runs in $O(nm)$ time.
We assume henceforth that $G$ is not an interval graph, so that $L(T) \geq 2$.
We apply Lemma~\ref{lem:expanded} and obtain an expanded tree model $\{T', \{T'_v\}_{v\in V_G}\}$ in $O(n^2)$ time.
As any host tree $T$ of $G$ has at most $n$ nodes \cite{ChaplickS14,HabibS09}, the expanded host tree $T'$ has $O(n)$ nodes by the third property of Lemma~\ref{lem:expanded}.
Moreover, observe that all subtrees $T'_v$ are rooted paths by the second property of Lemma~\ref{lem:expanded}.
Then we solve SFVS by computing $\optdir{V_{r}}{\emptyset}$ for the root $r$ of $T'$.

For this purpose, we construct a dynamic programming algorithm for computing $\optdir{V_r}{\emptyset}$.
The algorithm works on $T'$ in a bottom-up fashion, starting from the leaves and moving towards the root $r$.
As $T'$ is the host tree of an expanded tree model, there is a mapping between the vertices of $G$ and their corresponding root nodes in $T'$.
We start with defining the tables of data that the algorithm stores for each node $u \neq r$ of $T'$.
The constructed tables correspond to the sets $\optdir{V_{u}}{\emptyset}$, $\optdir{V_{u}}{\{w\}}$, $\optdir{V_{u,v}}{\emptyset}$, $\optdir{V_{u,v}}{\{w\}}$. 
In particular, we get the following formulas for the described sets.  
\begin{itemize}
\item Let $u,w\in V_{G}$ such that $u<w$ and $uw\in E_{G}$. Lemma~\ref{lem:together}~(i) and Lemma~\ref{lem:Adiruw} imply the following:
\begin{itemize}
\item If $u\in S$ or $w\in S$ then \\
$\optdir{V_{u}}{\{w\}}=\displaystyle\max_{weight}\left\{\bigcup_{u'\in\vartriangleleft u}{\optdir{V_{u'}}{\{w\}\cap N(u')}},\quad\{u\}\cup\bigcup_{u'\in\vartriangleleft uw}{\optdir{V_{u'}}{\{u,w\}\cap N(u')}}\right\}$.
\item If $u,w\notin S$ then \\
$\optdir{V_{u}}{\{w\}}=\displaystyle\max_{weight}\left\{\bigcup_{u'\in\vartriangleleft u}{\optdir{V_{u'}}{\{w\}\cap N(u')}},\quad\{u\}\cup(V_{\langle\vartriangleleft uw||\vartriangleleft u\rangle}\setminus S)\cup\bigcup_{u'\in\vartriangleleft uw}{\optdir{V_{u',u'\vartriangleleft u}}{\{u,w\}\cap N(u')}}\right\}$.
\end{itemize}

\item Let $u\in V_{G}$ and $v\in V_{G}\setminus S$ such that $u<v$ and $uv\in E_{G}$. Lemma~\ref{lem:together}~(i) and the description of $V_{u,v}\setminus\{v\}$ imply the following: \\
\hspace*{0.28in} $\displaystyle\optdir{V_{u,v}}{\emptyset}=\max_{weight}\left\{\optdir{V_{u,u\vartriangleleft v}}{\emptyset},\quad\{v\}\cup\optdir{V_{u,u\vartriangleleft v}}{\{v\}}\right\}.$

\item Let $u\in V_{G}$ and $v,w\in V_{G}\setminus S$ such that $u<v<w$ and $\{u,v,w\}$ induce a clique. Lemma~\ref{lem:together}~(i) and Lemma~\ref{lem:Adirutovandw} imply the following: \\
\hspace*{0.26in} $\displaystyle\optdir{V_{u,v}}{\{w\}}=\displaystyle\max_{weight}\left\{\optdir{V_{u,u\vartriangleleft v}}{\{w\}},\quad\{v\}\cup(V_{\langle\vartriangleleft vw|\{u\}|\{u\vartriangleleft v\}\rangle}\setminus S)\cup\bigcup_{u'\in V_{u}\cap\vartriangleleft vw}{\optdir{V_{u',u'\vartriangleleft v}}{\{v,w\}\cap N(u')}}\right\}$.
\end{itemize}

Let $v_{1},\ldots,v_{k}$ be the neighbors of $u$ on the path from $u$ towards the root $r$ of $T$ such that $v_{i}<v_{i+1}$, $1 \leq i<k$.
We compute $\optdir{V_{u}}{\emptyset}$ and $\optdir{V_{u}}{\{v_{i}\}}$, for each $1 \leq i \leq k$,
according to Lemma~\ref{lem:together}~(ii) and Lemma~\ref{lem:Adiruw}, respectively, by collecting the data stored on descendants of $u$.
%%% now it is a bit confusing...
%%% [S:] Is it?
For every $1\leq i<j\leq k$, with $v_{i},v_{j}\notin S$ and $v_{i}v_{j}\in E_{G}$, we compute $\optdir{V_{u,v_{i}}}{\{v_{j}\}}$ according to Lemma~\ref{lem:Adirutovandw}.
Observe that $\optdir{V_{u,v_{i}}}{\emptyset}$ and $\optdir{V_{u,v_{i}}}{\{v_{j}\}}$ are computed by table entries that correspond to values of $\optdir{V_{u',v'}}{\{w'\}}$ with $u'\leq u$, $v'<v_{i}$, and $v'<w'\leq v_{j}$.
When reaching the root $r$ of $T'$, it is enough to compute $\optdir{V_{r}}{\emptyset}$ by Lemma~\ref{lem:together}~(ii).

%The correctness of our algorithm follows from the lemmata of this section.
To evaluate the running time of the algorithm, we assume that the input graph is a connected rooted path graph having at least one cycle, so that $n \leq m$.
For this, observe that we can simply run our algorithm on each connected component and any tree has a trivial solution as it does not contain any $S$-cycle.
Now let us first determine the number of table entries required by our dynamic programming algorithm.
Consider the entries corresponding to $\optdir{X}{Y}$.
The sets $X$ are either $V_{u}$ or $V_{u,v}$ for some $u,v\in V_{G}$ such that $u<v$ and $uv\in E_{G}$ of which there are in total $n$ and $m$, respectively.
The sets $Y$ are either $\emptyset$ or $\{w\}$ for some $w\in V_{G}$ of which there are in total $n+1$.
We conclude that our table entries are $O(n(n+m))$.
Calculating a single entry requires to collect values of $O(n)$ entries. Those entries are determined via the vertex sets $\vartriangleleft v$, $\vartriangleleft vw$ and $V_{u}\cap\vartriangleleft vw$ and the vertices $u\vartriangleleft v$, which are precalculated.
Observe that these objects are also $O(n(n+m))$ in total.
To calculate $\vartriangleleft v$ and $u\vartriangleleft v$ we need only to traverse the host tree once for every $v\in V_{G}$.
As there are $O(n)$ nodes in $T'$, such a computation takes $O(n^2)$ time in total.
Similarly, to calculate $\vartriangleleft vw$ and $V_{u}\cap\vartriangleleft vw$ we need only to traverse the host tree once for every $v,w\in V_{G}$ such that $v<w$ and $vw\in E_{G}$.
Thus the total preprocessing time can be accomplished in $O(n^2+nm)$ time.
Therefore, the total running time of our algorithm is $O(n^{2}m)$.
\end{proof}

\section{Vertex leafage to cope with SFVS}\label{sec:UndirectedP}
%%% needs few explanation....
Due to Theorem~\ref{theo:leafage} and Corollary~\ref{cor:leafage}, it is interesting to ask whether our results can be further extended on larger classes of chordal graphs.
Here we consider graphs of bounded vertex leafage as a natural candidate towards such an approach.
However we show that \textsc{Subset Feedback Vertex Set} is NP-complete on undirected path graphs which are exactly the graphs of vertex leafage at most two.
%Undirected path graphs form a proper subclass of chordal graphs and they are exactly the intersection graphs of undirected paths in an undirected tree.
%For an undirected path graph $G$, the undirected tree of its intersection graph is called the \emph{tree model} for $G$.
In particular, we provide a polynomial reduction from the NP-complete \mcut problem.
Given a graph $G$, the \mcut problem is concerned with finding a partition
of $V(G)$ into two sets $A$ and $\overline{A}$ such that the number of edges with one
endpoint in $A$ and the other one in $\overline{A}$ is maximum among all partitions.
For two disjoint sets of vertices $X$ and $Y$, we denote by $E(X,Y)$ the set $\{\{x,y\}\mid x \in X, y \in Y\}$.
In such terminology, \mcut aims on finding a set $A \subseteq V(G)$ such that $|E(A, V(G) \setminus A) \cap E(G)|$ is maximum.
The \emph{cut-set} of a set of vertices $A$ is the set of edges of $G$ with exactly one endpoint in $A$, which is $E(A, V(G) \setminus A) \cap E(G)$.
% Then, the value |(S, S)| is the cut size of S.
% A maximum (cardinality) cut on G is a cut with the maximum size among all cuts on G. We denote the size of a
% maximum cut of G by mcs(G). Finally, the MaxCut problem is the problem of
% determining the size of the maximum cut
The \mcut problem is known to be NP-hard for general graphs \cite{Karp72} and remains NP-hard even when the input
graph is restricted to be a split or 3-colorable or undirected graph \cite{BodlaenderJ00}.
%Interestingly, \mcut remains NP-hard even on undirected graphs \cite{BodlaenderJ00}.
We mention that our reduction is based on \mcut on general graphs.

Towards the claimed reduction, for any graph $G$ on $n$ vertices and $m$ edges,
we will associate a graph $H_G$ on $12n^2 + 4n + 2m$ vertices.
First we describe the vertex set of $H_G$.
For every vertex $v \in V(G)$ we have the following sets of vertices:
\begin{itemize}
\item $X(v) = \{x_v^1, \ldots, x_v^{2n}\}$ and $\overline{X}(v) = \{\overline{x}_v^1, \ldots, \overline{x}_v^{2n}\}$,
\item $Y(v) = \{y_v^1, \ldots, y_v^{2n+1}\}$ and $\overline{Y}(v) = \{\overline{y}_v^1, \ldots, \overline{y}_v^{2n+1}\}$,
\item $Z(v) = \{z_v^1, \overline{z}_v^1, \ldots, z_v^{2n+1},\overline{z}_v^{2n+1}\}$, and
\item $E(v) = \{(v,x) \mid \{v,x\} \in E(G)\}$.
\end{itemize}
Observe that for every edge $\{u,v\} \in E(G)$ there are two vertices in $H_G$ that correspond to the ordered pairs $(u,v)$ and $(v,u)$.
We denote by $\overline{E}(v)$ the set of vertices $(x,v)$ of $H_G$ such that $\{x,v\}\in E(G)$.
The edge set of $H_G$ contains precisely the following:
\begin{itemize}
\item all edges required for the set $\bigcup_{v\in V(G)} (Y(v) \cup \overline{Y}(v) \cup E(v))$ to form a clique and
\item for every vertex $v \in V(G)$,
\begin{itemize}
\item {all elements of the sets} $E(X(v),Y(v))$, \ $E(\overline{X}(v),\overline{Y}(v))$, \ $E(X(v), E(v))$, \ and \ $E(\overline{X}(v), \overline{E}(v))$;
\item $\{x_v^i, x_v^{n+i}\}, \{\overline{x}_v^i, \overline{x}_v^{n+i}\}$ for each $i \in [n]$;
\item $\{y_v^j, z_v^{j}\}, \{y_v^j, \overline{z}_v^{j}\}, \{\overline{y}_v^j, z_v^{j}\}, \{\overline{y}_v^j, \overline{z}_v^{j}\}$ for each $j \in [2n+1]$.
\end{itemize}
\end{itemize}
%\todo[inline,color=red!40]{The above words were (and should be) there. We describe all bisets of the vertex set of $H_{G}$ that are contained in the edge set of $H_{G}$.}
This completes the construction of $H_G$.
Observe that the vertices of each pair $x_v^i, x_v^{n+i}$ and $\overline{x}_v^i, \overline{x}_v^{n+i}$ are true twins, whereas $z_v^{i}, \overline{z}_v^{i}$ are false twins.
An example of $H_G$ is given in Figure~\ref{fig:graphcut}.

\begin{figure}[!t]
\centering
\includegraphics[scale=0.92]{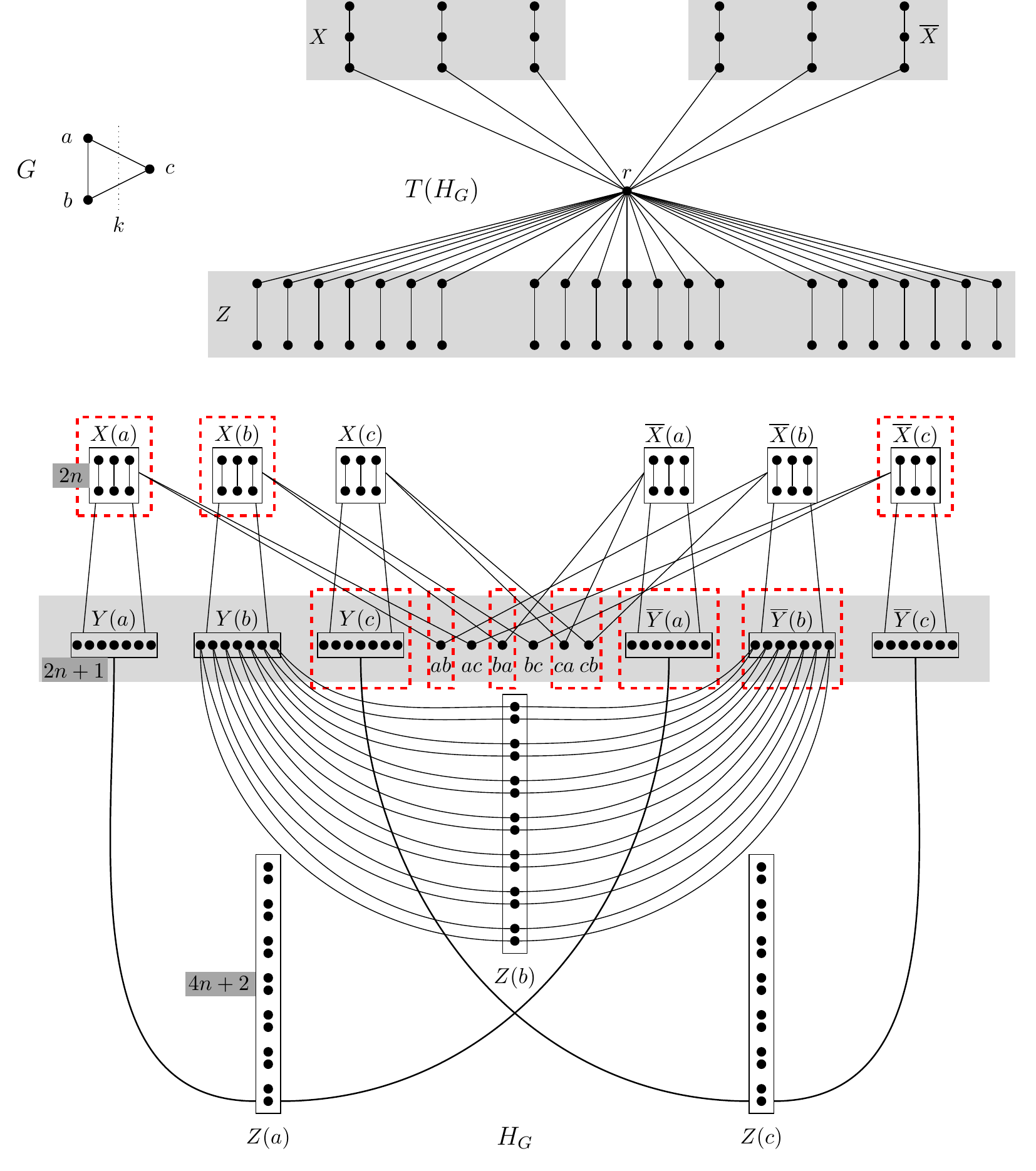}
\caption{Illustrating the undirected path graph $H_G$. On top left we show a graph $G$ on three vertices and on the bottom part we illustrate the corresponding graph $H_G$.
A tree model $T(H_G)$ for $H_G$, is given on the top right part.
The vertices of $H_G$ that lie on the grey area form a clique.}
%To keep the figure clean, all edges among the pairs $(v,X)$, $(w,\{z_1,z_2\})$, $(v',Y)$, and $(w', \cup \{z_1,z_2\})$ are assumed to be in $E(G) \setminus E(H)$.}
\label{fig:graphcut}
\end{figure}

%\footnote{We illustrate the construction of $H_G$ in Figure~\ref{fig:graphcut} given in Appendix~\ref{sec:example}.}.
%An example of $H_G$ is given in Appendix~\ref{sec:example}. in...Figure~\ref{fig:graphcut}.

%\begin{figure}[!t]
%\centering
%\includegraphics[scale=0.74]{figMC-all.pdf}
%\caption{Illustrating the undirected path graph $H_G$. On top left we show a graph $G$ on three vertices and on the bottom part we illustrate the corresponding graph $H_G$.
%A tree model $T(H_G)$ for $H_G$, is given on the top right part.
%%The vertices of $H_G$ that lie on the grey area form a clique.}
%%To keep the figure clean, all edges among the pairs $(v,X)$, $(w,\{z_1,z_2\})$, $(v',Y)$, and $(w', \cup \{z_1,z_2\})$ are assumed to be in $E(G) \setminus E(H)$.}
%\label{fig:graphcut}
%\end{figure}

%\begin{restatable}{lemma}{lemhgundirected}\label{lem:hg_undirected}
\begin{lemma}\label{lem:hg_undirected}
For any graph $G$, $H_G$ is an undirected path graph.
\end{lemma}
%\end{restatable}
\begin{proof}
In order to show that $H_G$ is an undirected path graph, we construct a tree model $T(H_G)$ for $H_G$ such that the vertices of $H_G$ correspond to particular undirected paths of $T(H_G)$.
To distinguish the vertex sets between $G$ and $T(H_G)$, we refer to the vertices of $T(H_G)$ as nodes.
%Let $X$ and $\overline{X}$ be the disjoint union of $n$ paths such that each path consists of $n$ vertices $x_1, \ldots, x_n$.
%Thus $X$ (and $\overline{X}$) contains exactly $n^2$ vertices.
%Moreover, let $Z$ be the disjoint union of $2n^2+n$ paths such that each path consists of two vertices $z_1,z_2$.
%% another try:
In order to construct $T(H_{G})$, starting from a particular node $r$, we create the following paths:
\begin{itemize}
\item for each $v \in V(G)$, $P_{X}(v) = \mypath{r, x_1^{(v)}, \ldots, x_n^{(v)}}$ and $P_{\overline{X}}(v) = \mypath{r, \overline{x}_1^{(v)}, \ldots, \overline{x}_n^{(v)}}$;
\item for each $v \in V(G)$ and $j \in [2n+1]$, $P_{Z}(v,j) = \mypath{r, z_1^{(v,j)}, z_2^{(v,j)}}$.
\end{itemize}
The tree model $T(H_G)$ is %constructed by attaching the above described paths on node $r$ and is obtained from
the union of the paths $P_{X}(v)$, $P_{\overline{X}}(v)$, and $P_{Z}(v,j)$ for all $v\in V(G)$ and for all $j\in[2n+1]$.

Next, we describe the undirected paths on $T(H_G)$ that correspond to the vertices of $H_G$.
\begin{itemize}
\item For every pair of vertices $x_v^{i}, x_v^{n+i} \in X(v)$ (resp., $\overline{x}_v^{i}, \overline{x}_v^{n+i} \in \overline{X}(v)$) with $i \in [n]$, we create two single-vertex paths $\mypath{x_i^{(v)}}$ (resp., $\mypath{\overline{x}_i^{(v)}}$).
\item For every pair of vertices $z_v^j,\overline{z}_v^j \in Z(v)$, with $j\in [2n+1]$, we create two single-vertex paths $\mypath{z_1^{(v,j)}}$ and $\mypath{z_2^{(v,j)}}$.
\item For every vertex $y_v^{j} \in Y(v)$ (resp., $\overline{y}_v^{j} \in \overline{Y}(v))$ with $j \in [2n+1]$,
we create a path $\mypath{x_n^{(v)}, \ldots, x_1^{(v)}, r, z_1^{(v,j)}, z_2^{(v,j)}}$
(resp., $\mypath{\overline{x}_n^{(v)}, \ldots, \overline{x}_1^{(v)}, r, z_1^{(v,j)}, z_2^{(v,j)}}$).
\item For every vertex $(u,v)$ of $H_G$, we create a path $\mypath{x_n^{(u)}, \ldots, x_1^{(u)}, r, \overline{x}_1^{(v)}, \ldots, \overline{x}_n^{(v)}}$.
\end{itemize}
Now it is not difficult to see that the intersection graph of the above constructed undirected paths of $T(H_G)$ is isomorphic to $H_G$.
More precisely, all paths containing node $r$ correspond to the vertices of the clique $\bigcup_{v\in V(G)} Y(v) \cup \overline{Y}(v) \cup E(v)$.
Moreover all subpaths of $P_{X}(v)$ and $P_{\overline{X}}(v)$ correspond to the vertices of $X(v)$ and $\overline{X}(v)$, respectively,
while the subpaths of $P_{Z}(v,j)$ correspond to the vertices of $Z(v)$.
Therefore, $H_G$ is an undirected path graph.
\end{proof}

Let us now show that there is a subset feedback vertex set in $H_G$ that is related to a cut-set in $G$.
We let $X=\bigcup X(v)$, $\overline{X} = \bigcup \overline{X}(v)$, and $Z = \bigcup Z(v)$.
%For doing so, we let $X=\bigcup X(v)$, $\overline{X} = \bigcup \overline{X}(v)$, and $Z = \bigcup Z(v)$.
%% Moreover, ...
%\begin{restatable}[*]{lemma}{lemundirectedforward}\label{lem:undirectedforward}
\begin{lemma}\label{lem:undirectedforward}
Let $G$ be a graph with $A \subseteq V(G)$ and
let $H_G$ be the undirected path graph of $G$. %and let $S = X \cup \overline{X} \cup Z$.
For the set of vertices $S = X \cup \overline{X} \cup Z$ of $H_G$,
there is a subset feedback vertex set $U$ of $(H_G,S)$ such that $|U| = 4n^2+n+2m-k$, where $k$ is the size of the cut-set of $A$ in $G$.
\end{lemma}
%\end{restatable}
\begin{proof}
We describe the claimed set $U$. For simplicity, we let $\overline{A} = V(G) \setminus A$ and
$\overrightarrow{E} = \{(u,v) \mid u\in A, v\in \overline{A}, \{u,v\}\in E(G)\}$. Clearly, $k=|\overrightarrow{E}|$.
\begin{itemize}
\item For every vertex $v \in A$, let $U(v) = X(v) \cup \overline{Y}(v)$.
\item For every vertex $v \in \overline{A}$, let $\overline{U}(v) = \overline{X}(v) \cup Y(v)$.
\item For every vertex $v \in V(G)$, let $E_U(v) = E(v) \setminus \overrightarrow{E}$.
\end{itemize}
Now $U$ contains all the above sets of vertices, that is, $U = \bigcup_{v\in V(G)} U(v) \cup \overline{U}(v) \cup E_U(v)$.
To show that $U$ is indeed a subset feedback vertex set, we claim that the graph $H_G - U$ does not contain any $S$-cycle.
Assume for contradiction that there is an $S$-cycle $C$ in $H_G - U$. Then $C$ contains at least one vertex from $S$.
We consider the following cases:
\begin{itemize}
\item Let $z \in C \cap Z$. Any vertex $z \in Z$ has exactly two neighbors $y,y'$ in $H_G$.
Thus $C$ must pass through both $y$ and $y'$. By construction, we know that $y \in Y(v)$ and $y' \in \overline{Y}(v)$ for a vertex $v \in V(G)$.
Hence we reach a contradiction to the fact that $y,y'$ are both vertices of $H_G - U$, since either $Y(v) \subseteq U$ and $\overline{Y}(v)\subseteq V(H_{G})\setminus U$ or vice versa.
\item Let $x \in C \cap X$. Observe that there exists $v\in \overline{A}$ such that $x \in X(v)$, since $X(v') \subseteq U$ for all $v' \in A$.
Also notice that all the vertices of $X(v)$ have exactly one neighbor in $X(v)$.
Thus there is a vertex $w \in (H_G - X(v))\cap C$ that is a neighbor of $x$ in $C$.
By construction, all vertices of $X(v)$ are adjacent to every vertex of $Y(v)$ and to every vertex of the form $(v,a)$ for which $\{v,a\} \in E(G)$.
Since $v \in \overline{A}$, we know that $Y(v) \subseteq U$ and $(v,a) \notin \overrightarrow{E}$ so that $(v,a) \in U$.
Hence, in both cases, we reach a contradiction to the fact that $w \in C$.
\item Let $x \in C \cap \overline{X}$.
%Similar to the previous case, there is a vertex $w \in (H_G - \overline{X}(v))\cap C$ that is a neighbor of $x$ in $C$, for a vertex $v \in A$.
%With similar arguments, we know that $w \notin \overline{Y}(v)$, which means the $w$ is a vertex of the form $(a,v)$.
%Since $v \in A$, we have $(a,v) \notin \overrightarrow{E}$, reaching
Arguments that are completely symmetrical to the ones employed in the previous case yield a contradiction to the fact that $w \in C$.
\end{itemize}
Therefore, there is no cycle that passes through a vertex of $S$, so that $U$ is indeed a subset feedback vertex set.
Regarding the size of $U$, observe that $|U(v)|=|\overline{U}(v)|=4n+1$ and $|\overrightarrow{E}| =k$.
Thus the described set $U$ fulfils the claimed properties.
\end{proof}

Now we are ready to show our main result of this section.
One direction follows from the previous lemma. The reverse direction is achieved through a series of claims.
The key part is to force at least $2m-k$ particular $S$-cycles of $H_G$ to appear within the corresponding $k$ cut-edges in $G$.

%\begin{restatable}[*]{theorem}{theonpundirected}\label{theo:npundirected}
\begin{theorem}\label{theo:npundirected}
Unweighted \textsc{Subset Feedback Vertex Set} is NP-complete on undirected path graphs.
\end{theorem}
%\end{restatable}
\begin{proof}
We provide a polynomial reduction from the NP-complete \mcut problem.
%% Note that we deal with the decision versions: the decision of \mcut problem takes as input a graph $G$ and an integer $k$, and asks whether there is a cut-set of size at least $k$.
Given a graph $G$ on $n$ vertices and $m$ edges for the \mcut problem, we construct the graph $H_G$.
Observe that the size of $H_G$ is polynomial and the construction of $H_G$ can be done in polynomial time.
By Lemma~\ref{lem:hg_undirected}, $H_G$ is an undirected path graph.
According to the terminology explained earlier for the vertices of $H_G$, we let  $S = X \cup \overline{X} \cup Z$.
We claim that $G$ admits a cut-set of size at least $k$ if and only if $(H_G,S)$ admits a subset feedback vertex set of size at most $4n^2+n+2m-k$.

Lemma~\ref{lem:undirectedforward} shows the forward direction.
The reverse direction is achieved through a series of claims.
In what follows, we are given a subset feedback vertex set $U$ of $(H_G,S)$ with $|U| \leq 4n^2+n+2m-k$.
%Let $y_v^j \in Y(v)$ for a vertex $v \in V(G)$ and $j \in [2n+1]$.
Based on the structure of $H_G$, it is not difficult to see that any $S$-triangle has one of the following forms for some $v\in V(G)$.
%% perhaps all S-cycles (?) %% here....
\begin{itemize}
\item $\mypath{x_v^i,x_v^{n+i},w}$ for some $i \in [n]$ and $w \in Y(v) \cup E(v)$;
\item $\mypath{x,w,w'}$ for some $x \in X(v)$ and $w,w' \in Y(v) \cup E(v)$;
\item $\mypath{\overline{x}_v^i,\overline{x}_v^{n+i},\overline{w}}$ for some $i \in [n]$ and $\overline{w} \in \overline{Y}(v) \cup \overline{E}(v)$;
\item $\mypath{\overline{x},\overline{w},\overline{w}'}$ for some $\overline{x} \in \overline{X}(v)$ and $\overline{w},\overline{w}' \in \overline{Y}(v) \cup \overline{E}(v)$;
\item $\mypath{y_v^j,\overline{y}_v^j,\tilde{z}}$ for some $j\in[2n+1]$ and $\tilde{z} \in \{z_v^j, \overline{z}_v^j\}$
\end{itemize}

\noindent We next define the following sets of vertices.

\begin{itemize}
%\begin{multicols}{2}
\item $\displaystyle\overleftrightarrow{E}=\bigcup_{v\in V(G)}(E(v)\cup\overline{E}(v))$,
\item $\displaystyle\overrightarrow{E}(A,B)=(\bigcup_{a\in A}E(a))\cap(\bigcup_{b\in B}\overline{E}(b))$ for every $A,B\subseteq V(G)$,
%\end{multicols}
%\begin{multicols}{2}
\item $A_{U}=\{v\in V(G)\mid Y(v)\subseteq V(H_{G})\setminus U\}$ for every $U\subseteq V(H_{G})$, and

\item $\overline{A}_{U}=\{v\in V(G)\mid\overline{Y}(v)\subseteq V(H_{G})\setminus U\}$ for every $U\subseteq V(H_{G})$.
%\end{multicols}
\end{itemize}
%For every $U\subseteq V(H_{G})$, we define $A_{U}=\{v\in V(G)\mid Y(v)\subseteq V(H_{G})\setminus U\}$ and $\overline{A}_{U}=\{v\in V(G)\mid\overline{Y}(v)\subseteq V(H_{G})\setminus U\}$.
Our task is to show that there exists a subset feedback vertex set $U$ such that $|U|\leq4n^2+n+2m-k$ and $U(A_{U})=U$, that is, $U$ satisfies the following properties:
\renewcommand\labelenumi{(\theenumi)}\begin{enumerate}
\item $Z\subseteq V(H_{G})\setminus U$.
\item For all $v\in V(G)$, either $X(v)\subseteq U$ or $X(v)\subseteq V(H_{G})\setminus U$ \ and\\
\hspace*{1.05in} either $\overline{X}(v)\subseteq U$ or $\overline{X}(v)\subseteq V(H_{G})\setminus U$.
\item For all $v\in V(G)$, either $X(v)\cup\overline{Y}(v)\subseteq U$ and $\overline{X}(v)\cup Y(v)\subseteq V(H_{G})\setminus U$\\
\hspace*{1.25in} or $\overline{X}(v)\cup Y(v)\subseteq U$ and $X(v)\cup\overline{Y}(v)\subseteq V(H_{G})\setminus U$.
\item $\overleftrightarrow{E}\setminus\overrightarrow{E}(A_{U},\overline{A}_{U})\subseteq U$ and $\overrightarrow{E}(A_{U},\overline{A}_{U})\subseteq V(H_{G})\setminus U$.
\end{enumerate}
For $i=0, 1, \ldots, 4$, we say that a set $U\subseteq V(H_{G})$ is a \emph{tier-$i$ sfvs} if it is a subset feedback vertex set of $(H_{G},S)$ that satisfies the first $i$ properties above.
In particular, any subset feedback vertex set of $(H_{G},S)$ is a tier-$0$ sfvs.
Notice that a tier-$i$ sfvs is a tier-$j$ sfvs for all $j\in\{0,1,\ldots,i-1\}$.
In such terms, our goal is to show that there exists a tier-$4$ sfvs $U$ such that $|U|\leq4n^2+n+2m-k$. %To that end, we make the following four Claims:
%
%\begin{restatable}{claim}{claimtieri}\label{claim:tieri}
%For every tier-$0$ sfvs $U$, there is a tier-$1$ sfvs $U'$ such that $|U'|\leq|U|$.
%\end{restatable}
%\begin{restatable}{claim}{claimtierii}\label{claim:tierii}
%For every tier-$1$ sfvs $U$, there is a tier-$2$ sfvs $U'$ such that $|U'|\leq|U|$.
%\end{restatable}
%\begin{restatable}{claim}{claimtieriii}\label{claim:tieriii}
%For every tier-$2$ sfvs $U$, there is a tier-$3$ sfvs $U'$ such that $|U'|\leq|U|$.
%\end{restatable}
%\begin{restatable}{claim}{claimtieriv}\label{claim:tieriv}
%For every tier-$3$ sfvs $U$, there is a tier-$4$ sfvs $U'$ such that $|U'|\leq|U|$.
%\end{restatable}
%
%We subsequently prove these Claims.

\begin{claim}\label{claim:tieri}
For every tier-$0$ sfvs $U$, there is a tier-$1$ sfvs $U'$ such that $|U'|\leq|U|$.
\end{claim}
\begin{claimproof}
If $Z \cap U = \emptyset$, then $U$ is already a tier-$1$ sfvs.
Assume that $\{z_{v}^{j},\overline{z}_{v}^{j}\}\cap U\neq\emptyset$ for some $v\in V(G)$ and some $j \in [2n+1]$.
We construct the set $U'=(U\setminus\{z_{v}^{j},\overline{z}_{v}^{j}\})\cup\{\overline{y}_{v}^{j}\}$.
Notice that $|\{z_{v}^{j},\overline{z}_{v}^{j}\}\cap U|\geq1$ and $|\{\overline{y}_{v}^{j}\}\setminus U|\leq1$, so that $|U'|\leq|U|$.
Also notice that $U'$ is a tier-$0$ sfvs because both vertices $z_v^{j}$ and $\overline{z}_v^{j}$ have at most one neighbor in $H_G - U'$.
Iteratively applying the same argument for each $v\in V(G)$ and for each $j\in[2n+1]$ such that $\{z_{v}^{j},\overline{z}_{v}^{j}\}\cap U\neq\emptyset$, we obtain a tier-$0$ sfvs $U'$ such that $Z\subseteq V(H_{G})\setminus U'$.
\end{claimproof}

\begin{claim}\label{claim:ybound}
For every tier-$1$ sfvs $U$, $|(Y(v) \cup \overline{Y}(v))\cap U|\geq 2n+1$ holds for every $v \in V(G)$.
\end{claim}
\begin{claimproof}
Consider the sets $Y(v)$ and $\overline{Y}(v)$ of a vertex $v \in V(G)$.
By the fact that $U$ is a $Z$-solution, we have $Z(v)\subseteq V(H_{G})\setminus U$.
Since $Z(v) \subset S$, the cycles $\mypath{y_{v}^{j},\overline{y}_{v}^{j},\overline{z}_{v}^{j}}$, $j\in[2n+1]$ are $2n+1$ vertex-disjoint $S$-cycles.
Therefore, the number of vertices of $Y(v) \cup \overline{Y}(v)$ in $U$ must be at least $2n+1$.
\end{claimproof}

\begin{claim}\label{claim:tierii}
For every tier-$1$ sfvs $U$, there is a tier-$2$ sfvs $U'$ such that $|U'|\leq|U|$.
\end{claim}
\begin{claimproof}
We consider the set $X(v)$ for a vertex $v \in V(G)$.
Assume that there are vertices $x \in X(v) \cap U$ and $x' \in X(v) \cap (V(H_G) \setminus U)$.
%Observe that the vertices of $X(v)$ have neighbors outside $X(v)$ only to the sets $Y(v) \cup E(v)$.
Observe that $N(x')\setminus X(v)=Y(v)\cup E(v)$.
%We claim that there is at most one vertex from $Y(v) \cup E(v)$ in $V(H_G) \setminus U$.
We claim that there is at most one vertex in $(Y(v)\cup E(v))\cap(V(H_G)\setminus U)$.
%Consider any two neighbors $w,w'$ of $x'$ in $Y(v) \cup E(v)$.
Consider any two vertices $w,w'\in Y(v)\cup E(v)$.
%Since $x' \in S$ and the vertices of $Y(v) \cup E(v)$ induce a clique, we have that both $w,w'$ are not vertices of $V(H_G) \setminus U$.
Since $x'\in S$ and $Y(v)\cup E(v)$ induces a clique, we have that at most one of $w,w'$ is in $V(H_G) \setminus U$.
%This means that there are at least $|Y(v) \cup E(v)|-1$ vertices of $Y(v) \cup E(v)$ in $U$.
%We construct the set $U' = U \setminus X(v) \cup \{w\}$, where $w \in (Y(v) \cup E(v)) \setminus U$.
We construct the set $U'=(U\setminus X(v))\cup(Y(v)\cup E(v))$.
Notice that $|X(v)\cap U|\geq1$ and $|(Y(v)\cup E(v))\setminus U|\leq1$, so  that $|U'|\leq|U|$.
Thus $Y(v) \cup E(v) \subseteq U'$ and $X(v) \cap U' = \emptyset$.
Moreover the constructed set $U'$ is indeed a tier-$1$ sfvs.
This follows from the fact that the subgraph induced by the vertices of $X(v)$ is acyclic and $N(X(v)) \subseteq U'$.
With completely symmetric arguments, there is a tier-$1$ sfvs $U''$ such that either $\overline{X}(v) \subseteq U''$ or $\overline{X}(v) \cap U'' = \emptyset$ holds, and $|U''|\leq |U'|$.
Iteratively applying these arguments for each $v\in V(G)$, we obtain a tier-$2$ sfvs $U'$ such that $|U'|\leq|U|$.
\end{claimproof}

\begin{claim}\label{claim:tieriii}
For every tier-$2$ sfvs $U$, there is a tier-$3$ sfvs $U'$ such that $|U'|\leq|U|$.
\end{claim}
\begin{claimproof}
{Consider a vertex $v\in V(G)$. By the fact that $U$ is a tier-$2$ sfvs, exactly one of the following holds:
\renewcommand\labelenumi{(\theenumi)}\begin{enumerate}
%\begin{multicols}{2}
\item $X(v)\cup\overline{X}(v)\subseteq V(H_{G})\setminus U$
\item $X(v)\subseteq U$ and $\overline{X}(v)\subseteq V(H_{G})\setminus U$
\item $X(v)\subseteq V(H_{G})\setminus U$ and $\overline{X}(v)\subseteq U$
\item $X(v)\cup\overline{X}(v)\subseteq U$
%\end{multicols}
\end{enumerate}}

Assume that (1) holds. Then $Y(v)\cup\overline{Y}(v)\subseteq U$ must hold, so the set $U'=(U\setminus Y(v))\cup X(v)$ is also a tier-$2$ sfvs and additionally $X(v)\cup\overline{Y}(v)\subseteq U'$ and $\overline{X}(v)\cup Y\subseteq V(H_{G})\setminus U'$ hold. Since $|Y(v)|=|X(v)|+1$, it follows that $|U'|<|U|$.

Now assume that (2) holds. Then $\overline{Y}(v)\subseteq U$ must hold, so the set $U'=U\setminus Y(v)$ is also a tier-$2$ sfvs and additionally $X(v)\cup\overline{Y}(v)\subseteq U'$ and $\overline{X}(v)\cup Y\subseteq V(H_{G})\setminus U'$ hold. For the case where (3) holds, completely symmetrical arguments yield that $U'=U\setminus\overline{Y}(v)$ is also a tier-$2$ sfvs and additionally $\overline{X}(v)\cup Y\subseteq U'$ and $X(v)\cup\overline{Y}(v)\subseteq V(H_{G})\setminus U'$ hold.

Lastly assume that (4) holds. By Claim \ref{claim:ybound} we have $|(Y(v)\cup\overline{Y}(v))\setminus U|\leq 2n+1$. Without loss of generality, assume that $|\overline{Y}(v)\setminus U|\leq n$. Then the set $U'=(U\setminus\overline{X}(v))\cup(\overline{Y}(v)\cup\overline{E}(v))$ is also a tier-$2$ sfvs and additionally $X(v)\cup\overline{Y}(v)\subseteq U'$ and $\overline{X}(v)\cup Y\subseteq V(H_{G})\setminus U'$ hold. Since $|\overline{X}(v)|=2n>|(\overline{Y}(v)\cup\overline{E}(v))\setminus U|$, it follows that $|U'|<|U|$.

Iteratively applying the arguments stated in the appropriate case above for each $v\in V(G)$, we obtain a tier-$3$ sfvs $|U'|$ such that $|U'|\leq|U|$.
\end{claimproof}

\begin{claim}\label{claim:tieriv}
For every tier-$3$ sfvs $U$, there is a tier-$4$ sfvs $U'$ such that $|U'|\leq|U|$.
\end{claim}
\begin{claimproof}
Let $U$ be a tier-$3$ sfvs. Clearly $(A_{U},\overline{A}_{U})$ is a partition of $V(G)$. Consider a vertex $e\in\overleftrightarrow{E}\setminus\overrightarrow{E}(A_{U},\overline{A}_{U})$. Then $e\in E(v)\cap\overline{E}(v')$ for some $v,v'\in V(G)$ such that $v\notin A_{U}$ or $v'\notin\overline{A}_{U}$, which implies that $X(v)\subseteq V(H_{G})\setminus U$ or $\overline{X}(v')\subseteq V(H_{G})\setminus U$. Since $\mypath{x_{v}^{0},x_{v}^{n},e}$ and $\mypath{\overline{x}_{v}^{0},\overline{x}_{v}^{n},e}$ are $S$-cycles, it follows that $e$ must be in $U$. Now consider a vertex $e\in\overrightarrow{E}(A_{U},\overline{A}_{U})$. Then $e\in E(v)\cap\overline{E}(v')$ for some $v\in A_{U}$ and some $v'\in\overline{A}_{U}$, which implies that $N(e)\cap S=X(v)\cup\overline{X}(v')\subseteq U$. It follows that the set $U'=U\setminus\overrightarrow{E}(A_{U},\overline{A}_{U})$ is a tier-$4$ sfvs.
%Since $N(e)\cap S=X(v)\cup\overline{X}(v')$, it follows that the set $U'=U\setminus\{e\}$ is also a tier-$3$ sfvs.
\end{claimproof}

To conclude our proof, let $U$ be an tier-$4$ sfvs such that $|U|\leq4n^{2}+n+2m-k$ that exists by Claim~\ref{claim:tieriv}.
By definition, for any $v\in V(G)$ exactly one of the following holds:
\begin{itemize}
\item $(X(v)\cup\overline{Y}(v))\setminus U=\emptyset$ and $(\overline{X}(v)\cup Y(v))\cap U=\emptyset$
\item $(\overline{X}(v)\cup{Y}(v))\setminus U=\emptyset$ and $({X}(v)\cup \overline{Y}(v))\cap U=\emptyset$.
\end{itemize}
This yields $|U\cap(X\cup\overline{X}\cup Y\cup\overline{Y}\cup Z)|=4n^{2}+n$, so that $|U\cap\overleftrightarrow{E}|\leq2m-k$.
Therefore we deduce that $|\overrightarrow{E}(A_{U},\overline{A}_{U})|\geq k$, which means that $A_{U}$ provides a desired cut-set in $G$.
\end{proof}

\section{Concluding Remarks}
We provided a systematic and algorithmic study towards the classification of the complexity of \textsc{Subset Feedback Vertex Set} on subclasses of chordal graphs.
We considered the structural parameters of leafage and vertex leafage as natural tools to exploit insights of the corresponding tree representation.
Our proof techniques revealed a fast algorithm for the class of rooted path graphs.
Naturally, it is interesting to settle whether the unweighted \textsc{Subset Feedback Vertex Set} problem is \FPT{} when parameterized by the leafage of a chordal graph.
Towards this direction, it is likely that the unweighted and weighted variants of the problem behave computationally different, as occurs in other cases \cite{abs-2007-14514,PapadopoulosT20}.
We also believe that our NP-hardness proof on undirected path graphs carries along the class of directed path graphs which are the intersection graphs of directed paths taken from an oriented tree (i.e., the underlying undirected graph is a tree).
%% It is natural to settle whether unweighted \textsc{Subset Feedback Vertex Set} is \FPT{} parameterized by the leafage of a chordal graph.
Further, in order to have a more complete picture on the behavior of the problem on subclasses of chordal graphs, strongly chordal graphs seems a candidate family of chordal graphs as they are incomparable to (vertex) leafage.

Moreover it would be interesting to consider the close related problem \textsc{Subset Odd Cycle Transversal} in which the task is to hit all odd $S$-cycles.
Preliminary results indicate that the two problems align on particular hereditary classes of graphs \cite{BJPPwg20,abs-2007-14514}.
As a byproduct, it is notable that all of our results for \textsc{Subset Feedback Vertex Set} are still valid for \textsc{Subset Odd Cycle Transversal}, as any induced cycle is an odd induced cycle (triangle) in chordal graphs.
More generally, an interesting direction for further research along the leafage is to consider induced path problems
that admit complexity dichotomies on interval graphs and split graphs, respectively \cite{BB86,HeggernesHLS15,IoannidouMN11,NatarajanS96}.

\section*{Acknowledgement}
We would like to thank Steven Chaplick for helpful remarks on a preliminary version.

\bibliography{subsetFVS_classes}

\begin{thebibliography}{10}

\bibitem{BergPTwg20}
Benjamin Bergougnoux, Charis Papadopoulos, and Jan~Arne Telle.
\newblock Node multiway cut and subset feedback vertex set on graphs of bounded
  mim-width.
\newblock In {\em Proceedings of WG 2020}, volume 12301, pages 388--400, 2020.

\bibitem{BB86}
Alan~A. Bertossi and Maurizio~A. Bonuccelli.
\newblock Hamiltonian circuits in interval graph generalizations.
\newblock {\em Information Processing Letters}, 23(4):195 -- 200, 1986.

\bibitem{BodlaenderJ00}
H.~L. Bodlaender and K.~Jansen.
\newblock On the complexity of the maximum cut problem.
\newblock {\em Nord. J. Comput.}, 7(1):14--31, 2000.

\bibitem{Bondy}
J.~A. Bondy and U.~S.~R. Murty.
\newblock {\em Graph Theory}.
\newblock Springer, 2008.

\bibitem{BH82}
Kellogg~S. {Booth} and J.~Howard {Johnson}.
\newblock {Dominating sets in chordal graphs}.
\newblock {\em {SIAM J. Comput.}}, 11:191--199, 1982.

\bibitem{BJPPwg20}
Nick Brettell, Matthew Johnson, Giacomo Paesani, and Dani{\"{e}}l Paulusma.
\newblock Computing subset transversals in h-free graphs.
\newblock In {\em Proceedings of WG 2020}, volume 12301, pages 187--199, 2020.

\bibitem{abs-2007-14514}
Nick Brettell, Matthew Johnson, and Dani{\"{e}}l Paulusma.
\newblock Computing weighted subset transversals in h-free graphs.
\newblock {\em CoRR}, abs/2007.14514, 2020.

\bibitem{B74}
P.~Buneman.
\newblock A characterization of rigid circuit graphs.
\newblock {\em Discrete Mathematics}, 9:205--212, 1974.

\bibitem{Chaplick19}
Steven Chaplick.
\newblock Intersection graphs of non-crossing paths.
\newblock In {\em Graph-Theoretic Concepts in Computer Science - 45th
  International Workshop, {WG} 2019}, volume 11789, pages 311--324, 2019.

\bibitem{ChaplickS14}
Steven Chaplick and Juraj Stacho.
\newblock The vertex leafage of chordal graphs.
\newblock {\em Discret. Appl. Math.}, 168:14--25, 2014.

\bibitem{fvs:chord:corneil:1988}
D.~G. Corneil and J.~Fonlupt.
\newblock The complexity of generalized clique covering.
\newblock {\em Discrete Applied Mathematics}, 22(2):109 -- 118, 1988.

\bibitem{CorneilP84}
Derek~G. Corneil and Yehoshua Perl.
\newblock Clustering and domination in perfect graphs.
\newblock {\em Discret. Appl. Math.}, 9(1):27--39, 1984.

\bibitem{CyganFKLMPPS15}
M.~Cygan, F.~V. Fomin, L.~Kowalik, D.~Lokshtanov, D.~Marx, M.~Pilipczuk,
  M.~Pilipczuk, and S.~Saurabh.
\newblock {\em Parameterized Algorithms}.
\newblock Springer, 2015.

\bibitem{Dietz84}
P.F. Dietz.
\newblock {\em Intersection graph algorithms}.
\newblock PhD thesis, Cornell University, 1984.

\bibitem{FellowsHRV09}
M.~R. Fellows, D.~Hermelin, F.~A. Rosamond, and S.~Vialette.
\newblock On the parameterized complexity of multiple-interval graph problems.
\newblock {\em Theor. Comput. Sci.}, 410(1):53--61, 2009.

\bibitem{FominHKPV14}
F.~V. Fomin, P.~Heggernes, D.~Kratsch, C.~Papadopoulos, and Y.~Villanger.
\newblock Enumerating minimal subset feedback vertex sets.
\newblock {\em Algorithmica}, 69(1):216--231, 2014.

\bibitem{FominGR20}
Fedor~V. Fomin, Petr~A. Golovach, and Jean{-}Florent Raymond.
\newblock On the tractability of optimization problems on h-graphs.
\newblock {\em Algorithmica}, 82(9):2432--2473, 2020.

\bibitem{GJ}
M.~R. Garey and D.~S. Johnson.
\newblock {\em Computers and Intractability}.
\newblock W. H. Freeman and Co., 1978.

\bibitem{G74}
F.~Gavril.
\newblock The intersection graphs of subtrees of trees are exactly the chordal
  graphs.
\newblock {\em Journal of Combinatorial Theory Series B}, 16:47--56, 1974.

\bibitem{Gavril75}
Fanica Gavril.
\newblock A recognition algorithm for the intersection graphs of directed paths
  in directed trees.
\newblock {\em Discret. Math.}, 13(3):237--249, 1975.

\bibitem{GolovachHKS14}
P.~A. Golovach, P.~Heggernes, D.~Kratsch, and R.~Saei.
\newblock Subset feedback vertex sets in chordal graphs.
\newblock {\em J. Discrete Algorithms}, 26:7--15, 2014.

\bibitem{HabibS09}
Michel Habib and Juraj Stacho.
\newblock Polynomial-time algorithm for the leafage of chordal graphs.
\newblock In {\em 17th Annual European Symposium, {ESA} 2009}, volume 5757,
  pages 290--300, 2009.

\bibitem{HeggernesHLS15}
Pinar Heggernes, Pim van~'t Hof, Erik~Jan van Leeuwen, and Reza Saei.
\newblock Finding disjoint paths in split graphs.
\newblock {\em Theory Comput. Syst.}, 57(1):140--159, 2015.

\bibitem{IoannidouMN11}
Kyriaki Ioannidou, George~B. Mertzios, and Stavros~D. Nikolopoulos.
\newblock The longest path problem has a polynomial solution on interval
  graphs.
\newblock {\em Algorithmica}, 61(2):320--341, 2011.

\bibitem{JaffkeKST19}
Lars Jaffke, O{-}joung Kwon, Torstein J.~F. Str{\o}mme, and Jan~Arne Telle.
\newblock Mim-width {III.} graph powers and generalized distance domination
  problems.
\newblock {\em Theor. Comput. Sci.}, 796:216--236, 2019.

\bibitem{JaffkeKT20}
Lars Jaffke, O{-}joung Kwon, and Jan~Arne Telle.
\newblock Mim-width {II.} the feedback vertex set problem.
\newblock {\em Algorithmica}, 82(1):118--145, 2020.

\bibitem{Karp72}
R.~M. Karp.
\newblock Reducibility among combinatorial problems.
\newblock In R.~E. Miller and J.~W. Thatcher, editors, {\em Complexity of
  Computer Computations}, pages 85--103. Plenum Press, 1972.

\bibitem{LinMW98}
In{-}Jen Lin, Terry~A. McKee, and Douglas~B. West.
\newblock The leafage of a chordal graph.
\newblock {\em Discuss. Math. Graph Theory}, 18(1):23--48, 1998.

\bibitem{Monma86}
Clyde~L. Monma and Victor~K. Wei.
\newblock Intersection graphs of paths in a tree.
\newblock {\em J. Comb. Theory, Ser. {B}}, 41(2):141--181, 1986.

\bibitem{NatarajanS96}
Sridhar Natarajan and Alan~P. Sprague.
\newblock Disjoint paths in circular arc graphs.
\newblock {\em Nord. J. Comput.}, 3(3):256--270, 1996.

\bibitem{Panda01}
B.~S. Panda.
\newblock The separator theorem for rooted directed vertex graphs.
\newblock {\em J. Comb. Theory, Ser. {B}}, 81(1):156--162, 2001.

\bibitem{PapT19}
Charis Papadopoulos and Spyridon Tzimas.
\newblock Polynomial-time algorithms for the subset feedback vertex set problem
  on interval graphs and permutation graphs.
\newblock {\em Discret. Appl. Math.}, 258:204--221, 2019.

\bibitem{PapadopoulosT20}
Charis Papadopoulos and Spyridon Tzimas.
\newblock Subset feedback vertex set on graphs of bounded independent set size.
\newblock {\em Theor. Comput. Sci.}, 814:177--188, 2020.

\bibitem{PhilipRST19}
Geevarghese Philip, Varun Rajan, Saket Saurabh, and Prafullkumar Tale.
\newblock Subset feedback vertex set in chordal and split graphs.
\newblock {\em Algorithmica}, 81(9):3586--3629, 2019.

\bibitem{Pietrzak03}
K.~Pietrzak.
\newblock On the parameterized complexity of the fixed alphabet shortest common
  supersequence and longest common subsequence problems.
\newblock {\em J. Comput. Syst. Sci.}, 67(4):757--771, 2003.

\bibitem{Spinrad03}
J.~P. Spinrad.
\newblock {\em Efficient Graph Representations}.
\newblock American Mathematical Society, Fields Institute Monograph Series 19,
  2003.

\bibitem{Yannakakis81a}
M.~Yannakakis.
\newblock Node-deletion problems on bipartite graphs.
\newblock {\em SIAM J. Comput.}, 10(2):310--327, 1981.

\end{thebibliography}

\end{document}